\documentclass[english,11pt]{article}
\usepackage{libertineRoman}
\usepackage{libertineMono}
\usepackage[T1]{fontenc}
\usepackage[latin9]{inputenc}
\usepackage{geometry}
\PassOptionsToPackage{vlined, boxruled}{algorithm2e}
\usepackage{algorithm2e}
\usepackage{amsmath}
\usepackage{amsthm}
\usepackage{amssymb}
\usepackage{graphicx}
\usepackage{soul}
\usepackage[dvipsnames]{xcolor}
\usepackage{mdframed}

\makeatletter

\def\highlight{1}

\theoremstyle{plain}
\newtheorem{thm1}{\protect\theoremname}[section]
\theoremstyle{remark}
\newtheorem{pthm1}[thm1]{\protect\ptheoremname}
\theoremstyle{plain}
\newtheorem{lem1}[thm1]{\protect\lemmaname}
\theoremstyle{plain}
\newtheorem{obs1}[thm1]{\protect\Observationname}
\theoremstyle{plain}
\newtheorem{inv1}[thm1]{\protect\Invariantname}
\theoremstyle{plain}

\theoremstyle{definition}
\newtheorem{defn1}[thm1]{\protect\definitionname}
\theoremstyle{plain}
\newtheorem{fact1}[thm1]{\protect\factname}
\theoremstyle{remark}
\newtheorem{rem1}[thm1]{\protect\remarkname}
\theoremstyle{plain}
\newtheorem{prop1}[thm1]{\protect\propositionname}
\ifx\proof\undefined
\newenvironment{proof}[1][\protect\proofname]{\par
	\normalfont\topsep6\p@\@plus6\p@\relax
	\trivlist
	\itemindent\parindent
	\item[\hskip\labelsep\scshape #1]\ignorespaces
}{%
	\endtrivlist\@endpefalse
}
\providecommand{\proofname}{Proof}
\fi

\if 1\highlight 
	\newenvironment{thm}[1][]{%
		\begin{mdframed}[nobreak=true,backgroundcolor=Aquamarine!60]\begin{thm1}[#1]%
			}{\end{thm1}\end{mdframed}%
	}

	\newenvironment{lem}[1][]{%
		\begin{mdframed}[nobreak=true,backgroundcolor=YellowGreen!60]\begin{lem1}[#1]%
		}{\end{lem1}\end{mdframed}%
	}
	
	\newenvironment{obs}[1][]{%
		\begin{mdframed}[nobreak=true,backgroundcolor=Salmon!60]\begin{obs1}[#1]%
			}{\end{obs1}\end{mdframed}%
	}

	\newenvironment{inv}[1][]{%
		\begin{mdframed}[nobreak=true,backgroundcolor=Salmon!60]\begin{inv1}[#1]%
			}{\end{inv1}\end{mdframed}%
	}

	\newenvironment{fact}[1][]{%
		\begin{mdframed}[nobreak=true,backgroundcolor=Salmon!60]\begin{fact1}[#1]%
			}{\end{fact1}\end{mdframed}%
	}

	\newenvironment{rem}[1][]{%
		\begin{mdframed}[backgroundcolor=Salmon!60]\begin{rem1}[#1]%
			}{\end{rem1}\end{mdframed}%
	}

	\newenvironment{pthm}[1][]{%
	\begin{mdframed}[nobreak=true,backgroundcolor=Salmon!60]\begin{pthm1}[#1]%
		}{\end{pthm1}\end{mdframed}%
	}

	\newenvironment{defn}[1][]{%
		\begin{mdframed}[innerbottommargin=0.1cm,innertopmargin=0.1cm,backgroundcolor=Apricot!60]\begin{defn1}[#1]%
			}{\end{defn1}\end{mdframed}%
	}
	
	\newenvironment{prop}[1][]{%
		\begin{mdframed}[backgroundcolor=Goldenrod!60]\begin{prop1}[#1]%
			}{\end{prop1}\end{mdframed}%
	}
	
	\expandafter\let\expandafter\oldproof\csname\string\proof\endcsname
	\let\oldendproof\endproof
	\renewenvironment{proof}[1][\proofname]{%
		\begin{mdframed}[backgroundcolor=lightgray!60]\oldproof[#1]%
		}{\oldendproof\end{mdframed}}
\else
	\newenvironment{thm}[1][]{%
		\begin{thm1}[#1]%
			}{\end{thm1}%
	}
	
	\newenvironment{lem}[1][]{%
		\begin{lem1}[#1]%
			}{\end{lem1}%
	}
	
	\newenvironment{obs}[1][]{%
		\begin{obs1}[#1]%
			}{\end{obs1}%
	}
	
	\newenvironment{inv}[1][]{%
		\begin{inv1}[#1]%
			}{\end{inv1}%
	}

	\newenvironment{rem}[1][]{%
		\begin{rem1}[#1]%
			}{\end{rem1}%
	}

	\newenvironment{pthm}[1][]{%
		\begin{pthm1}[#1]%
			}{\end{pthm1}%
	}

	\newenvironment{defn}[1][]{%
		\begin{defn1}[#1]%
			}{\end{defn1}%
	}
	
	\newenvironment{prop}[1][]{\begin{prop1}[#1]}{\end{prop1}}
\fi

\@ifundefined{showcaptionsetup}{}{%
	\PassOptionsToPackage{caption=false}{subfig}}
\usepackage{subfig}

\usepackage{float}
\floatstyle{boxed} 
\restylefloat{figure}

\SetKwInOut{Input}{Input}
\SetKwProg{Initialization}{Initialization}{}{}
\SetKwProg{Fn}{Function}{}{end}
\SetKwProg{EFn}{Event Function}{}{end}
\SetKw{And}{and}
\SetKw{Break}{break}
\SetKw{Continue}{continue}
\SetKw{Or}{or}

\SetKwFunction{UponRequest}{UponRequest}
\SetKwFunction{UponDeadline}{UponDeadline}
\SetKwFunction{UponCritical}{UponCritical}
\SetKwFunction{InitService}{InitService}
\SetKwFunction{ForwardTime}{ForwardTime}
\SetKwFunction{Charge}{Charge}

\SetCommentSty{mycommfont}

\SetFuncSty{myfuncsty}

\LinesNumbered \RestyleAlgo{boxruled}

\date{}

\makeatother

\makeatletter
\newcommand{\algorithmfootnote}[2][\footnotesize]{%
	\let\old@algocf@finish\@algocf@finish
	\def\@algocf@finish{\old@algocf@finish
		\leavevmode\rlap{\begin{minipage}{\linewidth}
				#1#2
		\end{minipage}}%
	}%
}
\makeatother

\usepackage{babel}
\providecommand{\Observationname}{Observation}
\providecommand{\corollaryname}{Corollary}
\providecommand{\definitionname}{Definition}
\providecommand{\lemmaname}{Lemma}
\providecommand{\Invariantname}{Invariant}
\providecommand{\propositionname}{Proposition}
\providecommand{\remarkname}{Remark}
\providecommand{\theoremname}{Theorem}
\providecommand{\ptheoremname}{Theorem}
\providecommand{\factname}{Fact}

\definecolor{darkred}{RGB}{200,0,0}

\begin{document}	

\global\long\def\rank{\ell}

\global\long\def\val{\operatorname{val}}

\global\long\def\ini{\operatorname{ini}}

\global\long\def\chrg{\operatorname{ptr}}

\global\long\def\true{\text{\sc true}}

\global\long\def\false{\text{\sc false}}

\global\long\def\alg{\text{\sc alg}}

\global\long\def\opt{\text{\sc opt}}

\global\long\def\FL{\text{\sc FL}}

\global\long\def\PCFL{\text{\sc PCFL}}

\global\long\def\PCSF{\text{\sc PCSF}}

\global\long\def\SF{\text{\sc SF}}

\global\long\def\SN{\text{\sc SN}}

\global\long\def\MC{\text{\sc MC}}

\global\long\def\MWC{\text{\sc MWC}}

\global\long\def\OP{\text{\sc ND}}

\global\long\def\PCOP{\text{\sc PCND}}

\global\long\def\NWSF{\text{\sc NWSF}}

\global\long\def\nul{\text{\sc null}}

\global\long\def\crit{\text{\sc crit}}

\global\long\def\conn{\text{\sc conn}}

\global\long\def\S{\mathcal{S}}

\global\long\def\T{\mathcal{T}}

\global\long\def\E{\mathcal{E}}

\global\long\def\P{\mathcal{P}}

\global\long\def\C{\mathcal{C}}


\title{Beyond Tree Embeddings -- a Deterministic Framework for Network Design with Deadlines or Delay}
\author{%
	\begin{tabular}{c}
		Yossi Azar\tabularnewline
		\texttt{\footnotesize{}azar@tau.ac.il}\tabularnewline
		{\small{}Tel Aviv University}\tabularnewline
	\end{tabular}\and%
	\begin{tabular}{c}
		Noam Touitou\tabularnewline
		\texttt{\footnotesize{}noamtouitou@mail.tau.ac.il}\tabularnewline
		{\small{}Tel Aviv University}\tabularnewline
\end{tabular}}
\maketitle

\begin{abstract}
	We consider network design problems with deadline or delay. All previous results for these models are based on randomized embedding of the graph into a tree (HST) and then solving the problem on this tree. We show that this is not necessary. In particular, we design a deterministic framework for these problems which is not based on embedding. This enables us to provide deterministic $\text{poly-log}(n)$-competitive algorithms for Steiner tree, generalized Steiner tree, node weighted Steiner tree, (non-uniform) facility location and directed Steiner tree with deadlines or with delay (where $n$ is the number of nodes). 
	
	Our deterministic algorithms also give improved guarantees over some previous randomized results. In addition, we show a lower bound of $\text{poly-log}(n)$ for some of these problems, which implies that our framework is optimal up to the power of the poly-log. Our algorithms and techniques differ significantly from those in all previous considerations of these problems.
\end{abstract}

\newpage{}

\section{Introduction}

In online minimization problems with deadlines, requests are released over a timeline. Each request has an associated deadline, by which it must be served by any feasible solution. The goal of an algorithm is to give a solution which minimizes the total cost incurred in serving the given requests. 

Another model, which generalizes the deadline model, is that of online problems with delay. In those problems, requests again arrive over a timeline. While requests no longer have a deadline, each pending request (i.e. a request which has been released but not yet served) incurs growing delay cost. The total cost of the algorithm is the cost of serving requests plus the total delay incurred over those requests; the delay cost thus motivates the algorithm to serve requests earlier. 

In this paper, we consider classic network design problems in the deadline/delay setting. In the classic (offline) setting of network design, one is given a graph of $n$ nodes and a set of connectivity requests (e.g. pairs of nodes to connect). The input contains a collection of elements (e.g. edges) with associated cost. A request is satisfied by any subset of elements which serves the connectivity request (e.g. a set of edges which connects the requested pair of nodes). A feasible solution for the offline problem is a set of elements which simultaneously satisfies all connectivity requests.

Such an offline network design problem induces an online problem with deadlines/delay as follows. The input graph is again given in advance. The requests, however, arrive over a timeline (with either a deadline or a delay function). At any point in time, the algorithm may choose to  \emph{transmit} an offline solution (i.e. a set of elements); each pending request that is served by the transmitted solution in the offline setting is served by this transmission in the online setting. In keeping with previous work on these problems, this paper considers the \emph{clairvoyant} model, in which the deadline of a request -- or its future accumulation of delay -- is revealed to the algorithm upon the release of the request. 

We next discuss such induced network design problems with deadlines/delay that have been previously considered. The usual solution for such problems is to randomly embed the general input into a tree, incurring a distortion to the metric space, then solving the problem on the resulting tree. In this paper, we present frameworks which \emph{bypass} this usual mode of work, enabling improved guarantees, generality and simplicity.

\paragraph{Steiner tree with deadlines/delay.}In this problem, requests are released on nodes of a graph with costs to the edges. Serving these requests comprises transmitting a subgraph which connects the request and a designated root node of the graph. This problem was studied in the case in which the graph is a tree -- in this case it is called the \textbf{multilevel aggregation problem} (first presented in \cite{DBLP:conf/esa/BienkowskiBBCDF16}). With $D$ the depth of the input tree, the best known results for multilevel aggregation are $O(D)$ competitiveness for the deadline model by Buchbinder \textit{et al.}~\cite{DBLP:conf/soda/BuchbinderFNT17}, and $O(D^2)$ competitiveness for the delay model in ~\cite{DBLP:conf/focs/AzarT19}. Thus, a simple algorithm for general Steiner tree with deadlines/delay based on metric tree embedding for this problem is to embed a general graph into a tree, and then using the best multilevel aggregation algorithms; in both the deadline and delay case, this can be seen to yield $O(\log^2 n)$-competitive randomized algorithms.

%

\paragraph{Facility location with deadlines/delay.}In this problem, presented in \cite{DBLP:conf/focs/AzarT19}, the input graph has weights to the edges and facility costs to the nodes. Requests arrive on the nodes of the graph, to be served by transmissions. A transmission consists of a set of facilities $U$, and a collection of pending requests $Q$. The transmission serves the requests of $Q$, and has a cost which is the sum of facility costs of the nodes in $U$, plus the sum of distances from each request of $Q$ to the closest facility in $U$. The best known algorithms for both the deadline and delay variants of this problem, also based on tree embedding, are randomized and $O(\log^2 n)$ competitive -- but apply only to the uniform problem, where the nodes' facility costs are identical.
%

%
%

This paper introduces a general deterministic framework for solving such network design problems on general graphs, with deadlines or with delay, which does not rely on tree embeddings. This framework obtains improved results to both previous problems, as well as new results for Steiner forest, nonuniform facility location, multicut, Steiner network, node-weighted Steiner forest and directed Steiner tree.

\subsection{Our Results}

We now state specifically our results for network design problems with deadlines/delay. Let $\E$ be the collection of elements in an offline network design problem. In this paper, we show the following results.
\begin{enumerate}
	\item If there exists a deterministic (randomized) $\gamma$-approximation for the offline network design problem which runs in polynomial time, then there exists an $O(\gamma \log |\E|)$-competitive deterministic (randomized) algorithm for the induced problem with deadlines, which also runs in polynomial time.
	
	\item If there exists a deterministic (randomized) $\gamma$-approximation for the \emph{prize-collecting} variant of the offline network design problem, then there exists an $O(\gamma \log |\E|)$-competitive deterministic (randomized) algorithm for the induced problem with delay, which also runs in polynomial time.
\end{enumerate}

Each of those results is obtained through designing a framework which encapsulates the given approximation algorithm.

We consider several network design problems on a graph of $n$ nodes, which are described in Subsection \ref{subsec:Intro_ConsideredProblems}. Plugging into our frameworks previously-known offline approximations (for either the original or prize-collecting variants) yields the results summarized in Table \ref{tab:Intro_ResultsTable}. Except for the algorithm for directed Steiner tree (which is randomized and runs in quasi-polynomial time due to the encapsulated approximation), all algorithms are deterministic and run in polynomial time.

\begin{table}[h!]
	\begin{center}
		\caption{Framework Applications}
		\label{tab:Intro_ResultsTable}
		\begin{tabular}{l|c|c} 
			& With Deadlines & With Delay\\
			\hline
			Edge-weighted Steiner forest & $O(\log n)$ & $O(\log n)$\\
			Multicut & $O(\log^2 n)$ & $O(\log^2 n)$\\
			Edge-weighted Steiner network & $O(\log n)$ & $O(\log n)$\\
			Node-weighted Steiner forest &$O(\log^2 n)$ & $O(\log^2 n)$\\
			Facility location (non-uniform)& $O(\log n)$ & $O(\log n)$\\			
			Directed Steiner tree & $O\left(\frac{\log^3 n}{\log \log n}\right)$ & ? \footnotemark\\			
			
		\end{tabular}
	\end{center}
\end{table}
\footnotetext{We could find no approximation result for prize-collecting directed Steiner tree. We conjecture that such an approximation algorithm exists which loses only a constant factor apart from the best approximation for the original offline problem, in which case we obtain an identical guarantee to the deadline case.}

Our frameworks improve on previous results in the following way:
\begin{enumerate}
	\item For Steiner tree with deadlines/delay, we give $O(\log n)$-competitive deterministic algorithms, while the best previously-known algorithms are randomized and $O(\log^2 n)$-competitive \cite{DBLP:conf/esa/BienkowskiBBCDF16,DBLP:conf/focs/AzarT19}. 
	
	\item For facility location with deadlines/delay, the best previously-known algorithms are randomized, $O(\log^2 n)$-competitive \cite{DBLP:conf/focs/AzarT19}, and apply only for the uniform case (where facilities have the same opening cost). We give $O(\log n)$-competitive, deterministic algorithms which apply also for the non-uniform case.
\end{enumerate}

For node-weighted Steiner forest and directed Steiner tree, our results are relatively close to the optimal solution -- in appendix we show an $\Omega(\sqrt{\log n})$ lower bound on competitiveness through applying the lower bound of \cite{DBLP:journals/corr/abs-1807-08543} for set cover with delay. As an information-theoretic lower bound, it applies for algorithms with unbounded computational power.

While the common regime in problems with deadlines/delay is that the number of requests $k$ is unbounded and the number of nodes $n$ is finite, we also address the opposite regime in which $k$ is small -- the latter being more popular in classic network design problems. We achieve the best of both worlds -- namely, we show a modification to the deadline/delay frameworks which replaces $n$ by $\min\{n,k\}$ in the competitiveness guarantees. This modification applies to all problems considered in this paper except for facility location, but conjecture that a similar algorithm would apply there as well.

\subsection{Our Techniques}
%

The \textbf{deadline framework} performs services (i.e. transmissions) of various costs; the logarithmic class of the cost of a service is called its level. Pending requests also have levels, which are maintained by the algorithm. Whenever a pending request of level $j$ reaches its deadline, a service of level $j+1$ starts. This service is only meant to serve requests of lower or equal level (we call such requests eligible for the service). After a service concludes, the level of remaining eligible requests is raised to that of the service. Intuitively, this means that once a pending request has seen a service of cost $2^j$, it refuses to be served by any cheaper service. This makes use of the aggregation property -- higher-cost services tend to be more cost-effective per request.

When a service is triggered, it has to choose which of the eligible requests to serve, subject to its budget constraint. The service prioritizes requests of earlier deadline, adding them until the budget is exceeded. The cost of serving those requests is estimated using the encapsulated approximation algorithm.

The main idea of levels exists in the \textbf{delay framework} as well. However, handling general delay functions requires more intricate procedures -- namely, for triggering a service and for choosing which requests to serve. The delay framework maintains an \emph{investment counter} for each pending request, which allows a service to pay for the delay of a request (i.e. the delay cost is charged to the budget of the service). A service is started when a large amount of delay for which no service has paid has accumulated on the requests of a particular level $j$ -- the started service is of level $j+1$. 

When choosing which of the eligible requests to serve, the algorithm considers the first point in time in which an eligible request would accumulate delay which is not paid for by its investment counter. Using its budget of $2^j$, it then attempts to push back this point in time farthest into the future -- it does so either by raising the investment counters, or by serving requests. The way to balance these two methods is problem-specific -- the framework thus formulates a prize-collecting instance, where the penalties represent future delay, and calls the encapsulated prize-collecting approximation algorithm to solve it.

\subsection{Considered Problems}
\label{subsec:Intro_ConsideredProblems}
In this paper, we consider the induced deadline/delay problems of several network design problems. We now introduce those problems.

\textbf{Steiner tree and Steiner forest.} In the Steiner forest problem, each request is a pair of terminals (i.e. nodes in the input graph), and the elements are the edges. A request is satisfied by a set of edges if the two terminals of the request are connected by those edges. The Steiner tree problem is an instance of Steiner forest in which the input also designates a specific node as the root, such that every request contains the root as one of its two terminals. A special case of the Steiner tree problem is the multilevel aggregation problem, in which the graph is a tree.

We also consider a stronger variant of the Steiner forest problem, in which each request is a \emph{subset} of nodes to be connected. While this problem is identical to the original Steiner forest in the offline setting (as the subset can be broken down to pairs), their induced deadline/delay problems are substantially different.

\textbf{Multicut.} In the offline multicut problem, each request is again a pair of terminals, and the elements are again the edges. A request is satisfied by a set of edges which, if removed from the original graph, would disconnect the pair of terminals. 

As in Steiner forest, it makes sense to define the stronger variant in which each request is a subset of nodes which must be disconnected from each other -- while both variants are equivalent in the offline setting, their induced deadline/delay problems are distinct.

\textbf{Node-weighted Steiner forest.} In this problem, the elements are the nodes, rather than edges. Each request is again a pair of terminals, and is satisfied by a solution which contains (in addition to the terminals themselves) nodes that connect the pair of terminals.

\textbf{Edge-weighted Steiner network.} This problem is identical to the Steiner forest problem, except that each request $q$ comes with a demand $f(q) \in \mathbb{N}$. A request is satisfied by a set of edges that contains $f(q)$ edge-disjoint paths between the terminals.

\textbf{Directed Steiner tree.} This problem is identical to the Steiner tree problem, except that the graph is now directed. Each pair request, where one of its terminals is the root, is satisfied by a set of edges that contain a directed path from the root to the other terminal.

\textbf{Facility location.} In the facility location problem, the requests are on the nodes of the graph. The elements are the nodes of the graph, upon which facilities can be opened. The cost of the solution is the total cost of opened facilities (opening cost) plus the distances from each request to the closest facility (connection cost).

The connection cost prevents facility location from being strictly compliant to the analysis of the framework we present. However, we nonetheless show that the framework itself applies to facility location as well.

\subsection{Related Work}

The classic online consideration of network design problems has been studied in numerous papers (e.g. \cite{DBLP:journals/siamdm/ImaseW91,DBLP:journals/algorithmica/Fotakis08,Berman1997,6108170,DBLP:journals/siamcomp/GuptaKR12,DBLP:journals/talg/AlonAABN06}). In this genre of problems, the connectivity requests arrive one after the other in a sequence (rather than over time), and must be served immediately by buying some elements which serve the request. These bought elements remain bought until the end of the sequence, and can thus be used to serve future requests. This is in contrast to the deadline/delay model considered in this paper, where the elements are \emph{transmitted} rather than bought, and thus future use of these elements requires transmitting them again (at additional cost). 

There is no connection between the classic online variant of a problem and the deadline/delay variant -- that is, neither problem is reducible to the other. There could be a stark difference in competitiveness between the two models, which depends on the network design problem. For some problems, the classic online admits much better competitive algorithms -- for example, in the multilevel aggregation problem, the classic online problem is Steiner tree on a tree, which is trivially $1$-competitive (while the best known algorithms for multilevel aggregation with deadlines/delay have logarithmic ratio). For other problems, the opposite is true -- for classic online directed Steiner tree, a lower bound of $\Omega(n^{1-\epsilon})$ exists on the competitiveness of any deterministic algorithm, for every $\epsilon>0$. In contrast, for directed Steiner tree with deadlines/delay, we present in this paper polylogarithmic-competitive algorithms.

The multilevel aggregation problem was first considered by Bienkowski \textit{et al.}~\cite{DBLP:conf/esa/BienkowskiBBCDF16}, who gave an algorithm with competitiveness which is exponential in the depth $D$ of the input tree, for the delay model. This result was then improved, first to $O(D)$ for the deadline model by Buchbinder \textit{et al.}~\cite{DBLP:conf/soda/BuchbinderFNT17}, and then to $O(D^2)$ for the general delay model in~\cite{DBLP:conf/focs/AzarT19}. These results yield $O(\log^2 n)$-competitive randomized algorithms for Steiner tree with deadlines/delay on general graphs, through metric embeddings; for more general Steiner problems (e.g. Steiner forest, node-weighted Steiner tree) no previously-known algorithm exists.

The multilevel aggregation also generalizes some past lines of work -- the TCP acknowledgement problem \cite{TCPAck_DBLP:conf/stoc/DoolyGS98,TCPAck_DBLP:journals/algorithmica/KarlinKR03,TCPAck_DBLP:conf/esa/BuchbinderJN07} is multilevel aggregation with $D=1$, and the joint replenishment problem \cite{JointRep_DBLP:conf/soda/BuchbinderKLMS08, JointRep_DBLP:journals/algorithmica/BritoKV12, JointRep_DBLP:conf/soda/BienkowskiBCJNS14} is multilevel aggregation with $D=2$.

Another problem studied in the context of delay is that of matching with delay \cite{Matching_DBLP:conf/approx/AshlagiACCGKMWW17,Matching_DBLP:conf/ciac/EmekSW17,Matching_DBLP:conf/stoc/EmekKW16,Matching_DBLP:conf/waoa/AzarF18,Matching_DBLP:conf/waoa/BienkowskiKLS18,Matching_DBLP:conf/waoa/BienkowskiKS17}. In this problem, requests arrive on points of a metric space, and gather delay until served. The algorithm may choose to serve two pending requests, at a cost which is the distance between those two requests in the metric space. This problem seems hard without making assumptions on the delay function, and thus is usually considered when the delay functions are identical and linear.

The $k$-server problem in the deadline/delay context has also been studied \cite{DBLP:conf/stoc/AzarGGP17,DBLP:conf/sirocco/BienkowskiKS18,DBLP:conf/focs/AzarT19}. In this problem, $k$ servers exist in a metric space, and requests again arrive on points of the space, gathering delay. To serve a request, the algorithm must move a server to that request, paying the distance between the server and the request. 

\section{Model and Deadline Framework}
\label{sec:FWD}

We are given a set $\E$ of elements, with costs $c:\E \to \mathbb{R}^+$. Requests are released over time, and we denote the release time of a request $q$ by $r_q$. Each request has a deadline $d_q$, by which it must be served. At any point in time, the algorithm may transmit a subset of elements $E\subseteq \E$, at a cost $\sum_{e\in E} c(e)$.

Each request $q$ is satisfied by a collection of subsets $X_q \subseteq 2^{\E}$ which is \emph{upwards-closed} -- that is, if $E_1\subseteq E_2 \subseteq \E$ and we have that $E_1 \in X_q$ then $E_2 \in X_q$. If the algorithm transmits the set of elements $E$, then all pending requests $q$ such that $E \in X_q$ are served by that transmission.

To give a concrete example of this abstract structure, consider the Steiner forest problem. In this problem, the elements $\E$ are the edges of a graph. For a request $q$ for the terminals $(u_1,u_2)$, the collection $X_q$ is the collection of edge sets $E'$ such that $(u_1, u_2)$ are in the same connected component in the spanning subgraph with edges $E'$.

One can also look at the corresponding offline problem -- given a set of requests $Q$, find a subset of elements $E'$ of the minimal total cost such that $E'\in X_q$ for every $q\in Q$. 

Now, consider a class of problems of this form -- such as Steiner tree for example -- and denote this class by $\OP$. The main result of this section is the following.

\begin{thm}
	\label{thm:FWD_Competitiveness}
	If there exists a $\gamma$ deterministic (randomized) approximation algorithm for $\OP$ which runs in polynomial time, then there exists an $O(\gamma\log |\E|)$-competitive deterministic (randomized) algorithm for $\OP$ with deadlines, which also runs in polynomial time.
\end{thm}

\begin{rem}
	\label{rem:FWD_QuasiPolynomialTime}
	If the approximation algorithm runs in \emph{quasi-polynomial} time, then the online algorithm also runs in quasi-polynomial time.
\end{rem}

\begin{rem}
	\label{rem:FWD_LasVegas}
	In this paper, we consider randomized approximation algorithms which have deterministic approximation guarantees and expected running time guarantees. Converting a randomized algorithm of \emph{expected} approximation guarantee and \emph{deterministic} running time to the format we consider can be achieved with repeated running of the algorithm until the resulting approximation is at most a factor of $2$ from the expected guarantee -- Markov's inquality ensures that the expected running time of this new algorithm is small.
	
	The only requirement for this conversion is that the algorithm is able to know whether its approximation meets the expected guarantee -- this requirement is met, for example, in all approximation algorithms based on LP solving + rounding (and in particular, all randomized algorithms in this paper).
\end{rem}

For a set of requests $Q$, we denote the solution for the offline problem returned by the $\gamma$ approximation by $\OP(Q)$. We also denote the optimal solution by $\OP^*(Q)$.

\subsection{The Framework}

We now present a framework for encapsulating an approximation algorithm for $\OP$ to obtain a competitive algorithm for $\OP$ with deadlines, thus proving Theorem \ref{thm:FWD_Competitiveness}.

\paragraph*{Calls to approximation algorithm.} The framework makes calls to the approximation algorithm for $\OP$ -- we denote such a call on a set of requests $Q$ by $\OP(Q)$ (the universe of elements $\E$, and the elements' costs, are identical to those of the online problem). Similarly, we denote the optimal solution for this set of requests by $\OP^*(Q)$. 

The framework also makes calls to $\OP$ where the costs of the elements are modified -- namely, that the cost of some subset of elements $E_0\subseteq \E$ is set to $0$. We use $\OP_{E_0 \gets 0}$ to denote such calls. 

When calling the approximation algorithm, we store the resulting solution (i.e. subset of elements) in a variable. If a solution is stored in a variable $S$, we use $c(S)$ to refer to the cost of that solution. Note that this cost is not necessarily the sum of costs of elements in that solution -- it is possible that the solution is for an instance in which the costs of some set of elements $E_0$ are set to $0$.

\paragraph*{Algorithm's description.} The framework is given in Algorithm \ref{alg:FWD_Algorithm}. For each pending request $q$, the algorithm maintains a level $\rank_q$. Upon the arrival of a new request $q$, the function $\UponRequest$ is called. This function assigns the initial value of the level of $q$, which is initially supposed to be the logarithmic class of the cost of the least expensive (offline) solution for $q$ -- the algorithm approximates this by making a call to the approximation algorithm on $\{q\}$, then dividing by the approximation ratio $\gamma$. Over time, the level of a request may increase.

Whenever a deadline of a pending request is reached, the function $\UponDeadline$ is called, and the algorithm starts a service. Services also have levels -- the level of a service $\lambda$, denoted by $\rank_\lambda$, is always $\rank_q +1$, where $q$ is the request which triggered the service. Intuitively, the service $\lambda$ is ``responsible'' for all pending requests of level at most $\rank_\lambda$ -- these requests are called the \emph{eligible} requests for $\lambda$. Overall, the service spends $O(\gamma \cdot 2^{\rank_\lambda})$ cost solely on serving these eligible requests.

The service constructs a transmission, which occurs at the end of the service. First, the service adds to the transmission all ``cheap'' elements -- those that cost at most $\frac{2^{\rank_\lambda}}{|\E|}$. Then, the service decides which of the eligible requests to serve, using the following procedure. It considers the requests by order of increasing deadline, adding them to the set of requests to serve. This process stops when either the cost of serving those requests, as estimated by the approximation algorithm, exceeds the budget ($O(\gamma \cdot 2^{\rank_\lambda})$), or the requests are all served. 

Since the amount by which the budget was exceeded in the ultimate iteration is unknown, the service transmits the solution found in the \emph{penultimate} iteration, in addition to a "singleton" solution to the last request to be served. 

The final step in the service is to ``upgrade'' the level of all eligible requests which are still pending after the transmission of the service. The level of those requests is assigned the level of the service. 

\begin{algorithm}
	\renewcommand{\algorithmcfname}{Algorithm}
	\caption{\label{alg:FWD_Algorithm} Network Design with Deadlines Framework} 
	
	\EFn{\UponRequest{$q$}}{
		
		Set $ S_q \gets \OP(\{q\})$
		
		Set $ I_q \gets \frac{c(S_q)}{\gamma} $.
		
		Set $\rank_q \gets \left\lfloor\log \left(I_q\right)\right\rfloor$ \tcp*[h]{the level of the request}
	}
	
	\BlankLine
	
	\EFn(\tcp*[h]{upon the deadline of a pending request $q$}){\UponDeadline{$q$}}{
		
		Start a new service $\lambda$, which we now describe.
		
		Set $\rank_\lambda \gets \rank_q + 1$. \label{line:FWD_SetServiceLevel}
		
		Set $Q_\lambda \gets \emptyset$.
		
		\BlankLine
		
		\tcp*[h]{buy all cheap elements}
		
		Set $E_0 \gets \left\{ e\in \E \middle| c(e) \le \frac{2^{\rank_\lambda}}{|\E|} \right\}$.	 	\label{line:FWD_BuyCheapEdges}
		
		\BlankLine
		
		\tcp*[h]{add eligible requests by order of deadline, until budget is exceeded}
		
		Set $S \gets \emptyset$.
		
		\While{there exists a pending $q' \notin Q_\lambda$ such that $\rank_{q^\prime} \le \rank_\lambda$}{\label{line:FWD_AddingRequestsToService}
			
			Let $q_{\text{last}} \notin Q_\lambda$ be the pending request with the earliest deadline such that $\rank_{q^\prime} \le \rank_\lambda$.
			
			Set $Q_\lambda \gets Q_\lambda \cup \{q_{\text{last}}\}$
			
			Set $S' \gets \OP_{E_0 \gets 0}(Q_\lambda)$.

			\lIf{$c(S') \ge \gamma\cdot 2^{\rank_\lambda}$}{\Break}\label{line:FWD_ServiceIsFull}
			
			Set $S \gets S'$.
			
		}
	
		\BlankLine
	
		Transmit the solution $E_0 \cup S \cup S_{q_{\text{last}}}$. \tcp*[h]{serve $Q_\lambda$} \label{line:FWD_ServeRequests}
		
		\BlankLine
		
		\tcp*[h]{upgrade still-pending requests to service's level}
		
		\ForEach{pending request $q^\prime$ such that $\rank_{q^\prime} \le \rank_\lambda$}{
	
			Set $\rank_{q^\prime} \gets \rank_\lambda $ \label{line:FWD_ServiceSetsRequestLevel}
			
		}
	}
\end{algorithm}

\subsection{Analysis}
To prove Theorem \ref{thm:FWD_Competitiveness}, we require the following definitions.

\subsubsection*{\emph{Definitions and Algorithm's Properties}}

Before delving into the proof of Theorem \ref{thm:FWD_Competitiveness}, we first define some terms used throughout the analysis, and prove some properties of the algorithm.

For a service $\lambda$, we call the value set to $\rank_{\lambda}$ the \emph{level} of $\lambda$; observe that this value does not change once defined. Similarly, for a request $q$, we call $\rank_q$ the level of $q$. Note that unlike services, the level of a request may change over time (more specifically, the level can be increased).

\begin{defn}[Service Pointer]
	Let $q$ be a request. We define $\chrg_q$ to be the last service $\lambda$ such that $\lambda$ sets $\rank_q \gets \rank_\lambda$ in Line \ref{line:FWD_ServiceSetsRequestLevel}. If there is no such service, we write $\chrg_q = \nul$. Similarly, we define $\chrg_q(t)$ to be the last service $\lambda$ before time $t$ such that $\lambda$ sets $\rank_q \gets \rank_\lambda$ in Line \ref{line:FWD_ServiceSetsRequestLevel} (with $\chrg_q(t)=\nul$ if there is no such service).
\end{defn}

\begin{defn}[Eligible Requests]
	\label{defn:FWD_EligibleRequest}
	Consider a service $\lambda$ and a request $q$ which is pending upon the start of $\lambda$, and has $\rank_q \le \rank_\lambda$ at that time. We say that $q$ was \emph{eligible} for $\lambda$.
\end{defn}

\begin{defn}[Types of Services]
	\label{defn:FWD_ServiceTypes}
	For a service $\lambda$, we say that:
	\begin{enumerate}
		\item $\lambda$ is \emph{charged} if there exists some future service $\lambda'$, which is triggered by a pending request $q$ reaching its deadline such that $\chrg_{q} (t_{\lambda'}) = \lambda$. We say that $\lambda'$ \emph{charged} $\lambda$.
		
		\item $\lambda$ is \emph{imperfect} if the $\Break$ command of Line \ref{line:FWD_ServiceIsFull} was reached in $\lambda$. Otherwise, we say that $\lambda$ is \emph{perfect}.
		
		\item $\lambda$ is \emph{primary} if, when triggered by the expired deadline of the pending request $q$, this request $q$ has $\chrg_q (t_\lambda) = \nul$. Otherwise, $\lambda$ is \emph{secondary}.
	\end{enumerate}
\end{defn}

A visualization of a possible set of services can be seen in Figure \ref{fig:FWD_ServiceDiagram}.

\begin{figure}[tb]
	\begin{center}
		\includegraphics{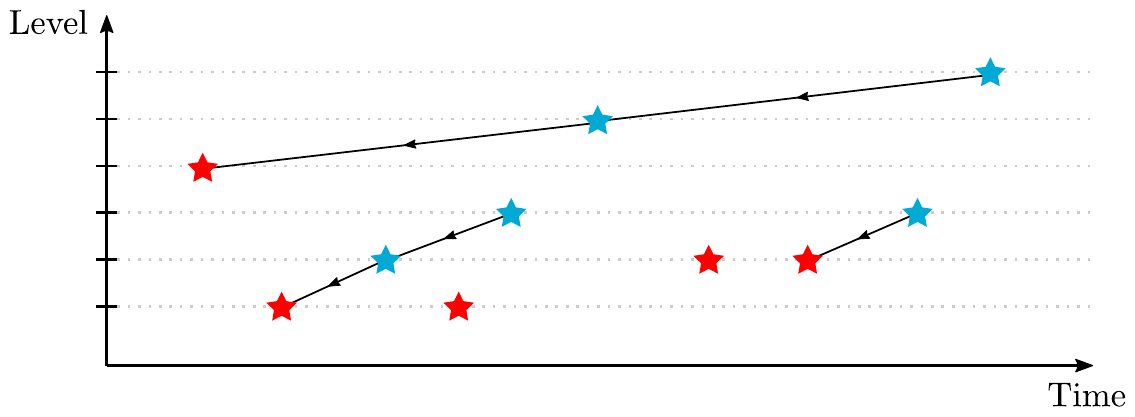}
	\end{center}

	This figure shows a possible set of services in a run of the algorithm. Each service is denoted by a star, where the location of the star indicates the time and level of the service. Primary services are denoted by red stars, and secondary services are denoted by blue stars. Each secondary service charges a previous service, of level one below its own; this charging is denoted by a directed edge from the secondary service to the charged service.
	
	Since every service can charge -- or be charged -- at most once,  the edges form disjoint paths. A property maintained by the algorithm is that a service ``dominates'' the quadrant of lesser-or-equal level and time -- once such a service occurs, no future secondary service would charge a service in this quadrant.
	
	\caption{\label{fig:FWD_ServiceDiagram} Visualization of Services}
\end{figure}

Fix any input set of requests $Q$. We denote by $\Lambda$ the final set of services by the algorithm. For every service $\lambda \in \Lambda$, we denote by $Q_\lambda$ the set of requests served by $\lambda$ (this is identical to the final value of the variable $Q_\lambda$ in the algorithm). We define $c(\lambda)$ to be the cost of the service $\lambda$. For any subset $\Lambda' \subseteq \Lambda$, we also write $c(\Lambda')=\sum_{\lambda\in \Lambda'}c(\lambda)$. Note that $\alg = c(\Lambda)$. 

We denote the set of primary services made by the algorithm by $\Lambda_1$, and the set of secondary services by $\Lambda_2$, such that $\Lambda = \Lambda_1 \cup \Lambda_2$. We denote the set of charged services by $\Lambda^\circ$.

\begin{prop}
	\label{prop:FWD_UniqueCharge}
	Each service $\lambda \in \Lambda^\circ$ is charged by at most one service.
\end{prop}
\begin{proof}
	Assume for contradiction that $\lambda$ is charged by both $\lambda_1$ and $\lambda_2$, at times $t_1$ and $t_2$ respectively, and assume without loss of generality that $t_1 <t_2$. $\lambda_2$ charged $\lambda$ due to the pending request $q_2$, such that $\rank_{q_2} = \rank_\lambda$ and $\chrg_{q_2}(t_{\lambda_2}) = \lambda$. Note that $q_2$ was pending before both $\lambda$ and $\lambda_2$, and was thus pending before $\lambda_1$. But after $\lambda_1$, all pending requests are of level at least $\rank_{\lambda_1} = \rank_\lambda + 1$, in contradiction to having $\rank_{q_2} = \rank_\lambda$ immediately before $\lambda_2$.
\end{proof}

The following lemma we prove shows that for a set of requests which exist in the same time, the collection of charged services which serve them has at most one service from each level. 

\begin{defn}
	\label{defn:FWD_IntersectingSet}
	We say that a set of requests ${Q'}=\{q_{1},\cdots,q_{k}\}$ is \emph{intersecting} if there exists time $t$ such that $t\in [r_{q_i},d_{q_i}]$ for every $i\in\{1,\cdots,k\}$. We call $t$ an \emph{intersection time} of ${Q'}$.
\end{defn}

\begin{lem}
	\label{lem:FWD_UniqueClass}
	Let $Q'$ be an intersecting set of requests. Let $\Lambda_{Q'} \subseteq \Lambda^{\circ}$ be the set of charged services in which a request from $Q'$ is served. Then for every $j\in\mathbb{Z}$,
	there exists at most one service $\lambda\in \Lambda_{Q'}$ such that $\rank_{\lambda}=j$.
\end{lem}
\begin{proof}
	Assume for contradiction that there exists $j\in\mathbb{Z}$ for which
	there exist two distinct services $\lambda_{1},\lambda_{2}\in \Lambda_{Q'}$ such that $\rank_{\lambda_{1}}=\rank_{\lambda_{2}}=j$.
	Assume without loss of generality that $t_{\lambda_1}<t_{\lambda_2}$.
	In addition, let $q_{1}\in {Q'}$ be a request served by $\lambda_{1}$, and define
	$q_{2}\in {Q'}$ to be a request served by $\lambda_{2}$. Let $t$ be an intersection time of ${Q'}$.
	
	Since $\lambda_1$ is charged, there exists a request $q'$ which was pending at its deadline, triggering a service $\lambda'$, such that $\chrg_{q'}(t_{\lambda'})= \lambda_1$. From the definition of $\chrg_{q'}$, we have that $\rank_{q'} = \rank_\lambda$ at time $t_{\lambda'}$. Thus, the service $\lambda'$ must be of level exactly $j+1$. Also note that $q'$ was eligible for $\lambda_1$.  Consider the following two cases:
	
	\begin{enumerate}
		\item  $t_{\lambda'}>t_{\lambda_2}$. Since $q'$ was pending at $t_{\lambda_1}$ and at $t_{\lambda'}$, and since $t_{\lambda_1} <t_{\lambda_2} <t_{\lambda'}$, we have that $q'$ was pending at $t_{\lambda_2} $. Observe that $\rank_{q'} = \rank_{\lambda_1}$ at $t_{\lambda_2}$, since $\lambda_1$ occurred before $\lambda_2$. But this means that $q'$ was eligible for $\lambda_2$, but was not served (since it was pending at $t_{\lambda'}$). Thus, $\lambda_2$ set $\rank_{q'} \gets \rank_{\lambda_2}$ in Line \ref{line:FWD_ServiceSetsRequestLevel}, in contradiction to having $\chrg_{q'} (t_{\lambda'})= \lambda_1$.
		
		\item $t_{\lambda'} < t_{\lambda_2}$. Consider that since $\chrg_{q'} (t_{\lambda'})= \lambda_1$, we know that $q'$ was eligible for $\lambda_1$. The service $\lambda_1$ added eligible requests by order of increasing deadline, and thus we know that the deadline of $q'$ is after the deadline of $q_1$. We know that ${Q'}$ is an intersecting set of requests, and thus $r_{q_2} \le d_{q_1}$. Therefore, we have that $r_{q_2} < d_{q'} = t_{\lambda'} < t_{\lambda_2}$, and thus $q_2$ was pending at $t_{\lambda'}$. We know that $q_2$ was eligible for $\lambda_2$, and thus $\rank_{q_2} \le j$ at that time. But this contradicts the fact that after $\lambda'$, every pending request has level at least $\rank_{\lambda'} = j+1$. 
	\end{enumerate}
\end{proof}

We now move on to proving Theorem \ref{thm:FWD_Competitiveness}. The proof consists of upper-bounding the cost of the algorithm and lower-bounding the cost of the optimal solution.

\subsubsection*{\emph{Upper-bounding $\alg$}}

We prove the following lemma, which provides an upper bound on the cost of the algorithm.

\begin{lem}
	\label{lem:FWD_ALGUpperBound}
	$\alg \le O(\gamma) \cdot \left(\sum_{\lambda\in \Lambda_1} 2^{\rank_\lambda} + \sum_{\lambda\in \Lambda^\circ} 2^{\rank_\lambda}\right)$
\end{lem}

\begin{prop}
	\label{prop:FWD_ServiceCostBoundedByLevel}
	The total cost of a service $\lambda$ is at most $O(\gamma)\cdot 2^{\rank_\lambda}$.
\end{prop}
\begin{proof}
	The cost of the service $\lambda$ is the cost of the transmission in Line \ref{line:FWD_ServeRequests}. The cost of this transmission is at most the sum of the three following costs: $C(E_0)$, $c(S)$, and $c(S_{q_\textup{last}})$. The total cost of $E_0$, by definition of $E_0$, is at most $2^{\rank_\lambda}$. 
	
	The cost $c(S)$ is at most $\gamma \cdot 2^{\rank_\lambda}$. To see this, observe that the loop of Line \ref{line:FWD_AddingRequestsToService} either ends in the first iteration (in which case $S=\emptyset$  and the cost is zero), or continues for two or more iterations. In the second case, consider the iteration before last -- since we did not break out of the loop, we have that $c(S) \le \gamma \cdot 2^{\rank_\lambda}$.
	
	As for the cost $c(S_{q_\textup{last}})$, consider the initial level of $q_{\textup{last}}$. Levels only increase over time, and we know that upon the service $\lambda$ we had that $\rank_{q_{\textup{last}}} \le \rank_\lambda$. Thus, the initial level of $q_{\textup{last}}$ was at most $\rank_\lambda$. According to the way in which the initial level is set, we thus have that $c(S_{q_\textup{last}}) \le 2\gamma\cdot 2^{\rank_\lambda}$.
	
	Summing over the three costs completes the proof.
\end{proof}

\begin{prop}
	\label{prop:FWD_OnlyFullServicesCharged}
	Only imperfect services can be charged.
\end{prop}

\begin{proof}
	Observe that a perfect service serves all eligible requests. Thus, Line \ref{line:FWD_ServiceSetsRequestLevel} is not called in such a service, which implies that the service is not charged.
\end{proof}

\begin{proof}[Proof of Lemma \ref{lem:FWD_ALGUpperBound}]
	Observe that $\alg = c(\Lambda_1) + c(\Lambda_2)$. First, observe that through Proposition \ref{prop:FWD_ServiceCostBoundedByLevel} we have that $c(\Lambda_1) \le O(\gamma) \cdot \sum_{\lambda\in \Lambda_1} 2^{\rank_\lambda}$. 
	
	It remains to show that $c(\Lambda_2)\le O(\gamma)\cdot \sum_{\lambda\in \Lambda^\circ} 2^{\rank_\lambda}$. Observe that every secondary service $\lambda$ of level $j$ charges a previous service $\lambda'\in \Lambda^\circ$ of level $(j-1)$. From Proposition \ref{prop:FWD_ServiceCostBoundedByLevel}, we have that $c(\lambda) \le O(\gamma)\cdot 2^j$, and thus $c(\lambda) \le O(\gamma)\cdot 2^{\rank{\lambda'}}$. Summing over all secondary services completes the proof, where Proposition \ref{prop:FWD_UniqueCharge} guarantees that no charged service is counted twice.
\end{proof}

\subsubsection*{\emph{Lower-bounding $\opt$}}
Fix the optimal solution for the given input, which consists of the services $\Lambda^*$ made in various points in time. Denote by $\opt$ the cost of this optimal solution. To complete the proof of Theorem \ref{thm:FWD_Competitiveness}, we require the following two lemmas which lower-bound the cost of the optimal solution. 

\begin{lem}
	\label{lem:FWD_PrimaryLowerBoundsOPT}
	$\sum_{\lambda\in \Lambda_1} 2^{\rank_\lambda} \le O(1)\cdot \opt$
\end{lem}

\begin{lem}
	\label{lem:FWD_ChargeLowerBoundsOPT}
	$\sum_{\lambda\in \Lambda^\circ} 2^{\rank_\lambda} \le O(\log |\E|) \cdot \opt$
\end{lem}

\begin{proof}[Proof of Lemma \ref{lem:FWD_PrimaryLowerBoundsOPT}]
	Observe that two primary services $\lambda_1,\lambda_2$ of the same level are triggered by two requests $q_1,q_2$ which are disjoint -- i.e. $[r_{q_1}, d_{q_1}] \cap [r_{q_2}, d_{q_2}] = \emptyset$. Otherwise, if $q_1$ and $q_2$ are not disjoint, then without loss of generality assume that $d_{q_1}\in [r_{q_2}, d_{q_2}]$. In this case, $\lambda_1$ would consider $q_2$, which is eligible (as $q_1,q_2$ are of the same level). This would either lead to $\lambda_1$ serving $q_2$, or $\chrg_{q_2}(t_{\lambda_2})\neq \nul$, both of which are contradictions to $\lambda_2$ being primary.
	
	Therefore, the requests triggering primary services of any specific level form a set of disjoint intervals. Now, let $m_j$ be the number of primary services of level $j$, and let $j_{\max}$ be the maximum level of a primary service. Denoting $x^+ = \max(x,0)$, we have that 
	\begin{align*}
	\sum_{\lambda\in \Lambda^1} 2^{\rank_\lambda} & = \sum_{j=-\infty}^{j_{\max}} m_j \cdot 2^j \\
	& \le   \sum_{j=-\infty}^{j_{\max}} \left(m_j - \max_{j'>j}\{m_{j'}\} \right)^+ \cdot 2^{j+1} \\
	& = 4 \cdot \sum_{j=-\infty}^{j_{\max}} \left(m_j - \max_{j'>j}\{m_{j'}\} \right)^+ \cdot 2^{j-1} 
	\end{align*}
	where the inequality is through changing the order of summation and summing a geometric series.
	
	Now, consider the optimal solution. For each primary service $\lambda$ triggered by a request $q$, we know that $\rank_{q} = \rank_\lambda -1$, and that $\chrg_{q}(t_\lambda) = \nul$. Thus, $\rank_\lambda -1$ was the initial level of $q$, set in $\UponRequest$. Thus, we have that $\OP^*(\{q\})\ge \frac{\OP(\{q\})}{\gamma} \ge 2^{\rank_\lambda-1}$. 
	
	This implies that the optimal solution must create $m_{j_{\max}}$ services of cost at least $2^{j_{\max}-1}$ each, to serve the (disjoint) requests which trigger level $j_{\max}$ primary services. In addition, the optimal solution must create at least $(m_{j_{\max}-1} - m_{j_{\max}})^+$ \emph{additional} services, of cost at least $2^{j_{\max}-2}$ each, to service requests that trigger level $(j_{\max} -1)$ primary services. Repeating this argument, for each level $j$ the optimal solution must pay an additional cost of $\left(m_j - \max_{j'>j}\{m_{j'}\} \right)^+ \cdot 2^{j-1}$. Overall, we have that 
	
	\[\opt \ge \sum_{j=-\infty}^{j_{\max}} \left(m_j - \max_{j'>j}\{m_{j'}\} \right)^+ \cdot 2^{j-1}\]

	 and thus $\sum_{\lambda\in \Lambda_1} 2^{\rank_\lambda} \le 4\cdot \opt$.
\end{proof}

It remains to prove Lemma \ref{lem:FWD_ChargeLowerBoundsOPT}, i.e. charging $2^{\rank_\lambda}$ for each service $\lambda \in \Lambda^\circ$ to the optimal solution times $O(\log |\E|)$. To do this, we split this charge of $2^{\rank_\lambda}$ between the services of the optimal solution. Proposition \ref{prop:FWD_LocalCharge} shows that this charge is valid.

For a service $\lambda^*\in \Lambda^*$ made by the optimal solution, denote the set of requests served in  $\lambda^*$ by $Q_{\lambda^*}$. Recall that for a service $\lambda \in \Lambda$ made by the algorithm, $Q_\lambda$ is the set of requests served by $\lambda$. For every $\lambda \in \Lambda$ and $\lambda^* \in \Lambda^*$, we define for ease of notation $Q_{\lambda \cap \lambda^*} \triangleq Q_\lambda \cap Q_\lambda^*$. 

For a set of requests $Q'$, we denote the cost of the optimal offline solution for $\OP$ on $Q'$ by $\OP^*(Q')$. We also use $\OP^*_{E_0 \gets 0}(Q')$ to refer to the cost of the optimal offline solution for $Q'$ where the costs of the elements $E_0 \subseteq \E$ is set to $0$. For a service $\lambda \in \Lambda$, we denote by $E_0^\lambda$ the value set to $E_0$ in Line \ref{line:FWD_BuyCheapEdges} during the service $\lambda$. The outline of the charging scheme is given in Figure \ref{fig:FWD_ChargingToOptimal}.

\begin{figure}[tb]
	\subfloat[\label{subfig:FWD_ChargingToOptimal_Charging}Charging Scheme]{\includegraphics[width=0.45\textwidth]{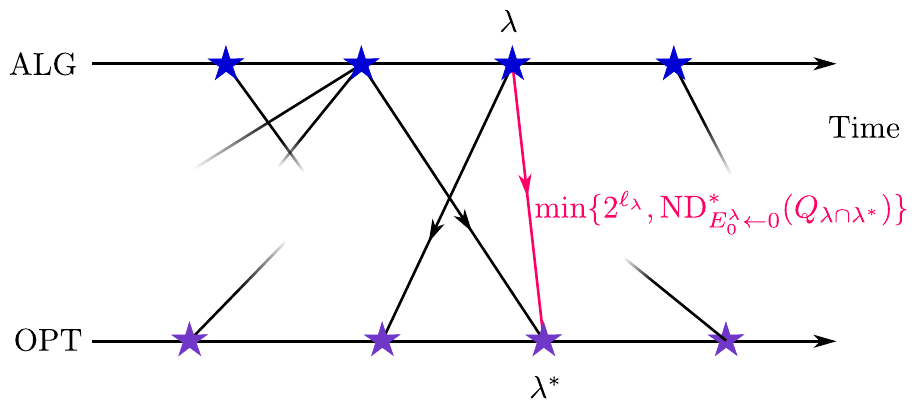}}
	\hfill
	\subfloat[\label{subfig:FWD_ChargingToOptimal_Sink}Charges to Optimal Service]{\includegraphics[width=0.45\textwidth]{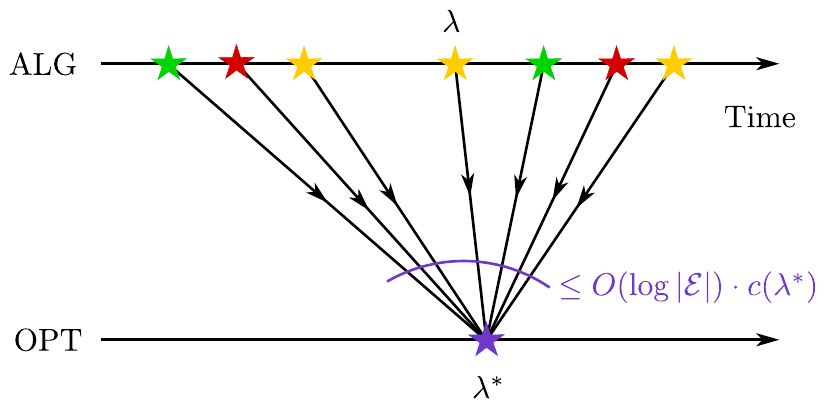}}\\

	\quad Subfigure \ref{subfig:FWD_ChargingToOptimal_Charging} shows the services of $\Lambda^\circ$ and the services of the optimal algorithm, as well as the charging of costs to the optimal solution. The amount $\min\{2^{\rank_\lambda}, \OP^*_{E_0^\lambda\gets 0}(q_{\lambda\cap\lambda^*})\}$ is charged by the service $\lambda \in \Lambda^\circ$ to the optimal service $\lambda^*$. In the proof of Lemma \ref{lem:FWD_ChargeLowerBoundsOPT}, we show that these charges are sufficient, i.e. each service $\lambda\in \Lambda^\circ$ charges at least $2^{\rank_\lambda}$.
	
	\quad Subfigure \ref{subfig:FWD_ChargingToOptimal_Sink} shows the validity of the charging, given in Proposition \ref{prop:FWD_LocalCharge}. This proposition shows that the total amount charged to an optimal service $\lambda^*$ exceedes its cost by a factor of at most $O(\log |\E|)$. This is shown by partitioning the services which charge cost to $\lambda^*$ into three types. The first type (green) is low-level services, which are shown to charge a total of at most $O(1)\cdot c(\lambda^*)$. The second type (yellow) is medium-level services. Each of these charges at most $c(\lambda^*)$, but there are at most $O(\log |\E|)$ such yellow services. The last type (red), high-level services, are shown to charge $0$ to $\lambda^*$. 
	
	\caption{\label{fig:FWD_ChargingToOptimal} Visualization of Services}
\end{figure}

\begin{prop}
	\label{prop:FWD_LocalCharge}
	There exists a constant $\beta$ such that for every optimal service $\lambda^* \in \Lambda^*$, we have that 
	\begin{equation}
		\label{eq:FWD_ChargePerOptimalService}
		\sum_{\lambda\in \Lambda^\circ}\min\{2^{\rank_\lambda},\OP^*_{E_0^\lambda\gets 0}(Q_{\lambda \cap \lambda^*})\} \le \beta\log |\E| \cdot c(\lambda^*)
	\end{equation}
\end{prop}

\begin{proof}
	Fix an optimal service $\lambda^* \in \Lambda^*$. Denote by $\Lambda'\subseteq \Lambda^\circ$ the subset of charged services made by the algorithm in which a request from $Q_{\lambda^*}$ is served (other services, for which $Q_{\lambda \cap \lambda^*}=\emptyset$, need not be considered). Observe that $Q_{\lambda^*}$ is an intersecting set, as the optimal solution served $Q_{\lambda^*}$ is a single point in time. Lemma \ref{lem:FWD_UniqueClass} implies that for every level $j$, there exists at most one $j$-level service in $\Lambda'$. Define $\ell = \lfloor \log (c(\lambda^*)) \rfloor$.  Now, consider the following cases for a service $\lambda\in \Lambda'$:
	\begin{enumerate}
		\item $\rank_\lambda \le \ell$. Each such $\lambda$ contributes at most $2^{\rank_\lambda}$ to the left-hand side of Equation \ref{eq:FWD_ChargePerOptimalService}. Summing over at most one service from each level yields a geometric sum which is at most $2^{\ell +1} \le 2\cdot c(\lambda^*)$.
		
		\item $\ell < \rank_\lambda < \ell + \lceil \log |\E|  \rceil+ 1$. For such $\lambda$, observe that $\min\{2^{\rank_\lambda}, \OP^*_{E_0^\lambda \gets 0}(Q_{\lambda \cap \lambda^*})\} \le \OP^*(Q_\lambda) \le c(\lambda^*)$. Summing over at most a single service from each level, the total contribution to the left-hand side of Equation \ref{eq:FWD_ChargePerOptimalService} from these levels is at most $\lceil \log |\E| \rceil\cdot c(\lambda^*)$.
		
		\item $\rank_\lambda \ge \ell + \lceil \log |\E| \rceil +1$. Observe that $\min\{2^{\rank_\lambda}, \OP^*_{E_0^\lambda \gets 0}(Q_{\lambda \cap \lambda^*})\} \le \OP^*_{E_0^\lambda \gets 0}(Q_{\lambda^*})$. We now claim that $\OP^*_{E_0^\lambda \gets 0}(Q_\lambda^*) =0$, which implies that the total contribution from these levels to the left-hand side of Equation \ref{eq:FWD_ChargePerOptimalService} is $0$. 
		
		Indeed, consider that every element in $\lambda^*$ costs at most $c(\lambda^*) \le 2^{\ell +1}$. Thus, since $2^{\rank_\lambda} \ge 2^{\ell+1} \cdot |\E|$, we have that $\lambda$ added all elements of $\lambda^*$ to $E_0^\lambda$ in Line \ref{line:FWD_BuyCheapEdges}. Thus, $\lambda^*$ is itself a feasible solution for $Q_{\lambda^*}$ of cost $0$, completing the proof.
	\end{enumerate}
	Summing over the contributions from each level completes the proof.
\end{proof}

\begin{proof}[Proof of Lemma \ref{lem:FWD_ChargeLowerBoundsOPT}]
	It is enough to show that for every charged service $\lambda \in \Lambda^\circ$, we have that
	\begin{equation}
		\label{eq:FWD_GlobalChargeIsLocalCharge}
		 2^{\rank_\lambda} \le \sum_{\lambda^* \in \Lambda^*} \min\{2^{\rank_\lambda},\OP^*_{E_0^\lambda\gets 0}(Q_{\lambda \cap \lambda^*})\} 
	\end{equation}
	
	Summing over all $\lambda \in \Lambda^\circ$ and using Proposition \ref{prop:FWD_LocalCharge} would immediately yield the lemma.
	
	If one of the summands on the right-hand side of Equation \ref{eq:FWD_GlobalChargeIsLocalCharge} is $2^{\rank_\lambda}$, the claim clearly holds, and the proof is complete. Otherwise, the right-hand side is exactly $\sum_{\lambda^* \in \Lambda^*}\OP^*_{E_0^\lambda\gets 0}(Q_{\lambda \cap \lambda^*}) $. Observe that $\bigcup_{\lambda^* \in \Lambda^*} Q_{\lambda \cap \lambda^*} = Q_\lambda$, and thus a feasible solution for $Q_\lambda$ is to take the union of the elements of the optimal solutions for $Q_{\lambda \cap \lambda^*}$ for every $\lambda^*$. This implies that 
	\[  \OP^*_{E_0^\lambda\gets 0}(Q_\lambda)   \le \sum_{\lambda^* \in \Lambda^*}\OP^*_{E_0^\lambda\gets 0}(Q_{\lambda \cap \lambda^*}) \]
	
	We claim that $2^{\rank_\lambda}   \le \OP^*_{E_0^\lambda\gets 0}(Q_\lambda)$, which completes the proof. Indeed, from Proposition \ref{prop:FWD_OnlyFullServicesCharged}, we know that $\lambda$ is an imperfect service. This means that during the construction of $Q_\lambda$, the loop of Line \ref{line:FWD_AddingRequestsToService} was completed in the $\Break$ command of Line \ref{line:FWD_ServiceIsFull}. Observing the value of the variable $S'$ at that line, we have that $c(S') \ge \gamma\cdot 2^{\rank_\lambda}$. Since $S'$ was obtained from a call to $\OP_{E_0^\lambda \gets 0}(Q_\lambda)$, the guarantee of the approximation algorithm for $\OP$ implies that $\OP^*_{E_0 \gets 0}(Q_\lambda) \ge 2^{\rank_\lambda}$.
\end{proof}

\begin{proof}[Proof of Theorem \ref{thm:FWD_Competitiveness}]
	The competitiveness of the algorithm results immediately from Lemmas \ref{lem:FWD_ALGUpperBound}, \ref{lem:FWD_PrimaryLowerBoundsOPT} and \ref{lem:FWD_ChargeLowerBoundsOPT}.
	
	As for the running time, it is clear that the main cost of the algorithm is calling the approximation algorithm $\OP$, and that this is done $O(|Q|)$ times (every iteration of the loop in Line \ref{line:FWD_AddingRequestsToService} adds a request to the ongoing service).
\end{proof}

\section{Applications and Extensions of the Deadline Framework}
\label{sec:APP}

In this section, we apply the framework to solving some network design problems in the deadline model, as well as describe some extensions of the framework.


\subsection{Edge-Weighted Steiner Tree and Steiner Forest}
\label{subsec:APP_SFD}
In this subsection, we consider the edge- weighted Steiner tree problem with deadlines. In this problem, we are given a (simple) graph $G=(V,E)$ of $n$ nodes, with a cost function $c:E \to \mathbb{R}^+$ on the edges. In addition, the input designates a node $\rho \in V$ as the \emph{root}. Requests arrive over time, each with an associated deadline, where each request is a terminal $u\in V$.

At any point in time, the algorithm may transmit some subset of edges $E'\subseteq E$, at a cost which is $\sum_{e\in E'} c(e)$. A pending request $q$ for a node $u\in V$ is considered served by this transmission if $u$ is in the same connected component as $\rho$ in the subgraph $G'=(V,E')$. 

A more general problem is the edge-weighted Steiner forest problem with deadlines. In this problem, we are again given a simple graph $G=(V,E)$ of $n$ nodes, and a cost function $c:E\to \mathbb{R}^+$ on the edges. Each request is now a pair of terminals $(u_1,u_2)\in V$. Again, the algorithm can transmit a subset of edges $E'$, paying $\sum_{e\in E'} c(e)$, and serving any pending request $q$ on $(u_1,u_2)$ such that $u_1,u_2$ are in the same connected component in $G'=(V,E')$. Observe that Steiner tree with deadlines is a special case of Steiner forest with deadlines where each requested pair contains the root $\rho$.

The Steiner forest with deadlines problem is a special case of the $\OP$ problem we described in Section \ref{sec:FWD}. The collection of elements in this case is the set of edges. For a request $q$ between two terminals $(u_1,u_2)$, the set $X_q$ of transmissions satisfying $q$ is the set of all transmissions $E'\subseteq E$ such that $u_1$ and $u_2$ are in the same connected component in the subgraph $(V,E')$.

We apply the framework of Section \ref{sec:FWD} to the Steiner forest with deadlines problem, thus obtaining an algorithm for both Steiner tree and Steiner forest with deadlines. The following theorem is due to Goemans and Williamson~\cite{Goemans1992}.

\begin{pthm}[\cite{Goemans1992}]
	\label{thm:APP_SFD_GoemansWilliamson}
	There exists a deterministic $2$-approximation for (offline) edge-weighted Steiner forest.
\end{pthm}

Plugging the algorithm of Theorem \ref{thm:APP_SFD_GoemansWilliamson} into the framework of Section \ref{sec:FWD}, and observing that $\log |E| \le 2\log n$, we obtain the following theorem.

\begin{thm}
	There exists an $O(\log n)$-competitive deterministic algorithm for edge-weighted Steiner forest with deadlines which runs in polynomial time.
\end{thm}

\subsubsection*{Strong Edge-Weighted Steiner Forest}

In the original Steiner forest problem (without deadlines), requesting pairs could be used to ensure connectivity between more than two nodes in the graph. Indeed, one could guarantee connectivity between $k$ nodes by releasing $k-1$ pair requests.

In the Steiner forest with deadlines problem, this is no longer the case. Since the transmissions serving the $k-1$ pair requests can occur in different times, there is no guarantee that there exists a point in time in which \emph{all} $k$ nodes are connected.

This motivates the \emph{strong} Steiner forest problem with deadlines, in which requests consist of \emph{subsets} of nodes which must be connected at the same time. The corresponding offline problem is still regular Steiner forest (since subset requests can be reduced to pair requests in the offline setting). Thus, we can apply the framework to the approximation algorithm of Goemans and Williamson~\cite{Goemans1992} as for the standard Steiner forest with deadlines, and obtain the following theorem.

\begin{thm}
	There exists an $O(\log n)$-competitive deterministic algorithm for strong edge-weighted Steiner forest with deadlines which runs in polynomial time.
\end{thm}

\subsection{Multicut}

In this subsection, we consider the multicut problem with deadlines. In this problem, we are again given a (simple) graph $G=(V,E)$ of $n$ nodes, with a cost function $c:E\to \mathbb{R}^+$ on the edges. Requests arrive over time, each with an associated deadline, where each request is a pair of terminals $\{u_1,u_2\}\in V$.

At any point in time, the algorithm may choose to momentarily disrupt a subset of edges $E'\subseteq E$, at a cost of $\sum_{e\in E'}c(e)$. A pending request $q$, which consists of the pair or terminals $\{u_1, u_2\}$, is served by this disruption if $u_1$ and $u_2$ are in two distinct connected components in the graph $G'=(V,E\backslash E')$. 

This problem is a special case of the $\OP$ problem we described in Section \ref{sec:FWD}. The collection of elements in this case is again the set of edges. For any request $q$ for a pair of terminals $\{u_1,u_2\}$, the set of satisfying transmissions $X_q$ is the collection of subsets of edges of the form $E'$ such that $u_1$ and $u_2$ are in two distinct connected components in the subgraph $(V, E\backslash E')$.

The following result is due to Garg \textit{et al.}~\cite{Garg1993}.

\begin{pthm}[\cite{Garg1993}]
	\label{thm:APP_MCD_GargEtAl}
	There exists a deterministic, polynomial-time, $O(\log n)$-approximation for multicut.
\end{pthm}

Plugging the approximation algorithm of Theorem \ref{thm:APP_MCD_GargEtAl} into the framework of Section \ref{sec:FWD}, and observing that $\log |E| \le 2\log n$, yields the following theorem.
\begin{thm}
	\label{thm:MCD_Competitiveness}
	There exists a deterministic $O(\log^2 n)$-competitive algorithm for multicut with deadlines which runs in polynomial time.
\end{thm}

\subsubsection*{Strong Multicut}

As was the case in Steiner forest, using pair requests in the original offline multicut problem could ensure disconnection between subsets of nodes, which is not the case for the deadline problem. This again motivates a strong version of multicut with deadlines, in which each request is a collection of nodes to be simultaneously disconnected from one another through disrupting some edges. 

As in the Steiner forest problem, the fact that these subset requests can be reduced in the offline case to pair requests allows us to use the approximation algorithm of Theorem \ref{thm:APP_MCD_GargEtAl} in the framework of Section \ref{sec:FWD}, yielding the following theorem.

\begin{thm}
	There exists an $O(\log^2 n)$-competitive deterministic algorithm for strong multicut with deadlines which runs in polynomial time.
\end{thm}

\subsection{Node-Weighted Steiner Forest}

The Steiner forest (and Steiner tree) problems have also been considered in the setting in which vertices, rather than edges, are bought. In this subsection, we apply the framework in this setting.

Formally, in the node-weighted Steiner forest with deadlines problem, we are given a graph $G=(V,E)$ such that $|V|=n$, and a cost function $c:V\to \mathbb{R}^+$ over the vertices. Each request $q$ is of two terminals $u_1,u_2 \in V$, and comes with an associated deadline. At any point in time, the algorithm may transmit a subset of vertices $V'\subseteq V$, at a cost of $\sum_{v\in V'} c(v)$. This transmission serves a pending request $q$ if $u_1$ and $u_2$ are in the same connected component in the subgraph induced by $V'$ (and in particular $u_1,u_2 \in V'$).

The node-weighted Steiner forest is a special case of the $\OP$ problem we described in Section \ref{sec:FWD}. The collection of elements in this case is the set of nodes. For a request $q$ for a pair of terminals $(u_1,u_2)$, the set of satisfying transmissions $X_q$ is the collection of node subsets $V'\subseteq V$ such that $u_1$ and $u_2$ are connected in the subgraph induced by $V'$.

We apply the framework of Section \ref{sec:FWD} to the node-weighted Steiner forest with deadlines problem, thus obtaining an algorithm for the node-weighted versions of both Steiner tree and Steiner forest with deadlines.

The following theorem is due, independently, to Bateni \emph{et al.}~\cite{DBLP:journals/siamcomp/BateniHL18} and Chekuri \emph{et al.}~\cite{DBLP:conf/approx/ChekuriEV12}.

\begin{pthm}[\cite{DBLP:journals/siamcomp/BateniHL18,DBLP:conf/approx/ChekuriEV12}]
	\label{thm:APP_NWSFD_BateniEtAl}
	There exists a polynomial-time, deterministic $O(\log n)$-approximation algorithm for node-weighted Steiner forest.
\end{pthm}

Applying the framework of Section \ref{sec:FWD} yields the following theorem.

\begin{thm}
	There exists an $O(\log^2 n)$-competitive deterministic algorithm for node-weighted Steiner forest with deadlines which runs in polynomial time.
\end{thm}

\subsection{Edge-Weighted Steiner Network}

The (edge-weighted) Steiner network problem with deadlines is identical to the Steiner forest with deadlines problem in Subsection \ref{subsec:APP_SFD}, except that every pair request $q$ on two terminals $u_1,u_2 \in V$ also has an associated demand $f(q) \in \mathbb{N}$. A transmission of edges $E'$ now serves a pending request $q$ if there exist $f(q)$ edge-disjoint paths from $u_1$ to $u_2$ in the graph $(V,E')$.

The edge-weighted Steiner network is again a special case of $\OP$. As in the Steiner forest, the elements are the edges of the graph. For each request $q$ for a pair of terminals $\{u_1,u_2\}$ with demand $f(q)$, the set of satisfying transmissions $X_q$ is the collection of subsets of edges $E'\subseteq E$ such that there exist $f(q)$ edge-disjoint paths from $u_1$ to $u_2$ in $(V,E')$.

The following Theorem is due to Jain~\cite{DBLP:journals/combinatorica/Jain01}.

\begin{pthm}[\cite{DBLP:journals/combinatorica/Jain01}]
	\label{thm:APP_SND_Jain}
	There exists a polynomial-time, deterministic, $2$-approximation for offline edge-weighted Steiner network.
\end{pthm}

Plugging the offline approximation algorithm of Theorem \ref{thm:APP_SND_Jain} into the framework of Section \ref{sec:FWD}, and again observing that $\log |E| \le 2\log n$, yields the following theorem.

\begin{thm}
	\label{thm:APP_SND_Competitiveness}
	There exists an $O(\log n)$-competitive deterministic algorithm for edge-weighted Steiner network with deadlines which runs in polynomial time.
\end{thm}

\subsection{Directed Steiner Tree}
In the directed Steiner tree problem with deadlines, we are given a (simple) directed graph $G=(V,E)$, costs $c:E\to \mathbb{R}^+$ to the edges and a designated root $\rho \in V$. Each request $q$ is a terminal $v\in V$. At any point in time, the algorithm may transmit a set of directed edges $E'\subseteq E$. A pending request $q$ for a terminal $v$ is served by this transmission if there exists a (directed) path from $\rho$ to $v$ in the subgraph $G'=(V,E')$.

This problem is also a special case of $\OP$ in the same way as the undirected Steiner tree. That is, the elements are the edges of the tree, and a set of edges $E'\subseteq E$ is in $X_q$, for a request $q$ of a terminal $v$, if there exists a directed path from $\rho$ to $v$ in the graph $(V,E')$. 

The following theorem is due to Grandoni \textit{et al.}~\cite{Grandoni:2019:OAA:3313276.3316349}.
\begin{pthm}[\cite{Grandoni:2019:OAA:3313276.3316349}]
	\label{thm:APP_DSD_Grandoni}
	There exists a randomized $O(\frac{\log^2 n }{\log \log n})$-approximation for directed Steiner tree, which runs in quasi-polynomial time (specifically, $O(n^{\log^5 n})$ time).
\end{pthm}

As a result of plugging the algorithm of Theorem \ref{thm:APP_DSD_Grandoni} into the framework of Section \ref{sec:FWD}, and again observing that $\log |E| \le 2\log n$, yields the following theorem.

\begin{thm}
	\label{thm:APP_DSD_Competitiveness}
	There exists a randomized $O(\frac{\log ^3 n}{\log \log n})$-competitive algorithm for directed Steiner tree with deadlines, which runs in quasi-polynomial time.
\end{thm}

\subsection{Facility Location}

In the facility location with deadlines problem, we are given a graph $G=(V,E)$, such that $|V|=n$. We are also given a facility opening cost $f:V\to \mathbb{R}^+$, and weights $w:E\to \mathbb{R}^+$ to the edges.  Requests arrive over time on the nodes of the graph, each with an associated deadline.

At any point in time, the algorithm may choose a node $v\in V$, open a facility at that node, and choose some subset of pending requests $Q'$ to connect to that facility. This action serves the pending requests of $Q'$. Immediately after performing this atomic action, the facility disappears. The total cost of this transmission is $f(v)$ (the opening cost of the facility) plus $\sum_{q\in Q'} \delta(v,q)$, where $\delta$ is the shortest-path metric on nodes induced by the edge weights $w$.

The set of elements in this case is the set of nodes $V$ (where buying a node means opening a facility at that node). Observe that facility location does \textbf{not} conform neatly to the $\OP$ structure of the problems addressed in our framework -- indeed, opening facilities does not immediately serve requests, and paying an additional connection cost is required. One could force the problem into the framework by adding the connections (i.e. shortest paths from a request to facility) as elements -- however, as each request requires a different connection, this would result in $\Theta(n|Q|)$ elements, where $Q$ is the set of requests. The resulting loss over the approximation algorithm in this case would be $\Theta(\log n +\log |Q|)$. 

Nevertheless, we show that the framework can be applied without any modification to the facility location problem, with only the facilities as elements, yielding the desired guarantee ($O(\log n)$ loss). In this subsection, we modify the necessary parts in the analysis of the framework in order to fit the facility location problem.

First, we consider a constant-approximation algorithm for the offline facility location problem. There are many such algorithms; the following is due to Jain and Vazirani~\cite{Jain:2001:AAM:375827.375845}. 

\begin{pthm}[\cite{Jain:2001:AAM:375827.375845}]
	\label{thm:APP_FLD_JainVazirani}
	There exists a polynomial-time, deterministic  $\gamma_{\FL}$-approximation for offline facility location, where $\gamma_{\FL}=3$.
\end{pthm}

In this subsection, we prove that plugging the approximation algorithm of Theorem \ref{thm:APP_FLD_JainVazirani} into the framework of Section \ref{sec:FWD} yields the following theorem. 

\begin{thm}
	\label{thm:APP_FLD_Competitiveness}
	There exists an $O(\log n)$-competitive deterministic algorithm for facility location with deadlines, which runs in polynomial time.
\end{thm}

\begin{rem}
	\label{rem:APP_FLD_SolutionNature}
	While the framework for facility location is the same as for $\OP$, an important remark must be made about the nature of facility location solutions.
	
	In the original framework for $\OP$, we hold solutions in variables, where a solution $S$ is a subset of the universe of elements $\E$. In facility location, a solution $S$ to $\FL(Q)$ (the offline facility location problem on the set of requests $Q$) is of different form -- $S$ contains a subset $F\subseteq \E = V$ of facilities to open, \emph{plus} a mapping $\phi:Q\to F$ from the input requests to the facilities of $F$, which determines the connection cost of the solution.
	
	The cost of the solution $S=(F,\phi)$, referred to as $c(S)$ in the framework, is now the opening cost $\sum_{v\in F} f(v)$ plus the connection cost $\sum_{q\in Q} \delta(q,\phi(q))$. As for transmissions in Line \ref{line:FWD_ServeRequests}, transmitting $E_0 \cup S \cup S_{q_{\text{last}}}$ refers to transmitting the facilities of $E_0$, $S$ and $S_q$, and connecting requests according to the mappings of $S$ and $S_q$.
\end{rem}

\subsubsection*{Analysis}

Consider that theorem \ref{thm:APP_FLD_Competitiveness} would result immediately if we could reprove Lemmas \ref{lem:FWD_ALGUpperBound}, \ref{lem:FWD_PrimaryLowerBoundsOPT} and \ref{lem:FWD_ChargeLowerBoundsOPT} for facility location with deadlines. The proofs of Lemmas \ref{lem:FWD_ALGUpperBound} and \ref{lem:FWD_PrimaryLowerBoundsOPT} go through in an identical way to the original framework. As for Lemma \ref{lem:FWD_ChargeLowerBoundsOPT}, the only change required is in the proof of Proposition \ref{prop:FWD_LocalCharge}. We now go over the necessary changes. 

\begin{proof}[Proof of Proposition \ref{prop:FWD_LocalCharge} for facility location]
	We use the notation defined in the original proof of Proposition \ref{prop:FWD_LocalCharge}. 
	
	Observe that the proof of the proposition goes through until the case analysis of each service $\lambda \in \Lambda'$. The two first cases (namely, that $\rank_\lambda \le \ell$ or $\ell < \rank_\lambda < \ell + \lceil \log |\E|  \rceil+ 1$) go through entirely.
	
	The difference is in the third case, in which $\rank_\lambda \ge \ell + \lceil \log |\E| \rceil +1$. As was the argument in the original proof, it holds that all facilities that were opened in $\lambda^*$ are also open in $\lambda$. Now, consider 
	that there exists a solution for $Q_{\lambda \cap \lambda^*}$ which connects each request to its facility in $\lambda^*$. Therefore, we have that $\OP_{E_0^\lambda \gets 0}(Q_{\lambda \cap \lambda^*})$ is at most the connection cost of the requests of $Q_{\lambda \cap \lambda^*}$ in $\lambda^*$. Summing over all services $\lambda$ of this class yields that the total contribution to the left-hand side of Equation \ref{eq:FWD_ChargePerOptimalService} is at most the connection cost incurred by the optimal solution in $\lambda^*$, which is at most $c(\lambda^*)$. 
	
	Combining this third case with the previous two cases completes the proof.
\end{proof}

\subsection{Exponential-Time Algorithms}
In online algorithms, one is often interested in the information-theoretic bounds on competitiveness, without limitations on running time. The framework of Section \ref{sec:FWD} supports such constructions -- plugging in the algorithm which solves the offline problem optimally yields the following theorem.

\begin{thm}
	There exists an $O(\log |\E|)$-competitive algorithm for $\OP$ with deadlines (with no guarantees on running time). In particular, there exists an $O(\log n)$competitive algorithm for all problems in this paper, where $n$ is the number of nodes in the input graph.
\end{thm}

\section{Delay Framework}
\label{sec:FWY}
We now consider the $\OP$ problem with delay. This problem is identical to the problem with deadlines, except that instead of a deadline, each request $q$ is associated with a continuous, monotone-nondecreasing delay function $d_q(t)$, which is defined for every $t$, and tends to infinity as $t$ tends to infinity (ensuring that every request must be served eventually).

The framework we present for problems with delay requires an approximation algorithm for the prize-collecting variant of the offline problem. In the prize-collecting $\OP$ problem, denoted $\PCOP$, the input is again a set of requests $Q$, and an additional penalty function $\pi:Q \to \mathbb{R}^+$. A solution is a subset of elements $E$ which serves some subset $Q'\subseteq Q$ of the requests. The cost of the solution is $\sum_{e\in E} c(e) + \sum_{q\in Q\backslash Q'} \pi(q)$ -- that is, the total cost of the elements bought plus the penalties for unserved requests.

\begin{thm}
	\label{thm:FWY_Competitiveness}
	If there exists a $\gamma$ deterministic (randomized) approximation algorithm for $\PCOP$ which runs in polynomial time, then there exists a $O(\gamma \log |\E|)$-competitive deterministic (randomized) algorithm for $\OP$ with delay, which runs in polynomial time.
\end{thm}

Note that Remarks \ref{rem:FWD_QuasiPolynomialTime} and \ref{rem:FWD_LasVegas} apply here as well.

\subsection{The Framework}

We now describe the framework for $\OP$ with delay. 

\paragraph*{Calls to the prize-collecting approximation algorithm.} The framework makes calls to the approximation algorithm $\PCOP$ for the prize-collecting problem. Such a call is denoted by $\PCOP(Q,\pi)$, where $Q$ is the set of requests and $\pi:Q\to \mathbb{R}^+$ is the penalty function. Some calls are made with the subscript $E_0 \gets 0$, for some subset of elements $E_0$. This notation means calling $\PCOP$ on the modified input in which the cost of the elements $E_0$ is set to $0$. The framework also makes calls to $\OP$, an approximation algorithm for the original (not prize-collecting) variant of $\OP$. This approximation algorithm is obtained through calling $\PCOP$ with penalties of $\infty$ for each request.

\paragraph*{Investment counter.} The algorithm maintains for each request $q$ an \emph{investment counter} $h_q$. Raising this counter corresponds to paying for delay (both past and future) incurred by the request $q$. When referring to the value of the counter at a point in time $t$, we write $h_q(t)$. 

\begin{defn}[Residual delay]
	We define the \emph{residual delay} of a pending request $q$ at time $t$ to be $\rho_q(t) = \max(0, d_q(t) - h_q(t))$. Intuitively, this is the amount of delay incurred by $q$ which no service has covered until time $t$. For a set of requests $Q$ pending at time $t$, we also define $\rho_Q(t) = \sum_{q\in Q} \rho_q(t)$.
\end{defn}

\begin{defn}[Penalty function $\pi_{t \to t'}$]
	At a time $t$, and for every future time $t'>t$, we define the penalty function $\pi_{t\to t'}$ on pending requests at time $t$ in the following way. For a request $q$ pending at time $t$, we have that $\pi_{t \to t'}(q)=\max(0, d_q(t')-h_q(t))$. Intuitively, the penalty for a request, as  evaluated at time $t$, is the future residual delay of the request if the algorithm does not raise its investment counter until time $t'$.
\end{defn}

As in the deadline framework, the delay framework assigns a level $\rank_q$ to each pending request $q$. 

\begin{defn}[Critical level]
	At any point during the algorithm, we say that a level $j$ becomes \emph{critical} if the total residual delay of requests of level at most $j$ reaches $2^j$.
\end{defn}

\paragraph*{Algorithm's description.} The framework is given in Algorithm \ref{alg:FWY_Algorithm}. The algorithm consists of waiting until any level $j$ becomes critical, and then calling $\UponCritical(j)$. Whenever a new request $q$ is released, the function $\UponRequest(q)$ is called. 

The algorithm maintains the level of each pending request $q$, denoted $\rank_q$. This level is initially the logarithmic class of the cost of the cheapest solution (i.e. set of elements) serving $q$ (in fact, the algorithm estimates this by calling the approximation algorithm $\OP$ and dividing by its approximation ratio). Over time, the level of a request may increase. 

When a level $j$ becomes critical, this triggers a service $\lambda$ of level $\rank_{\lambda} = j+1$. Intuitively, the service $\lambda$ is responsible for all pending requests of level at most $\rank_\lambda$ -- these are called the eligible requests for $\lambda$. The service first starts by raising the investment counters of eligible requests until they all have zero residual delay.

After doing so, the service observes the first point in the future in which such an eligible request has positive residual delay. The goal of the service is to push this point in time (called the forwarding time) as far into the future as possible, while spending at most $O(\gamma\cdot 2^{\rank_\lambda})$ cost. 

There are two methods of accomplishing this: the first is to raise the investment counters of the requests, and the second is serving the requests. The best course of action is to combine both methods in a smart manner -- deciding which eligible requests are to be served, and raising the investment counter for the remainder of the eligible requests. 

To achieve this, the service finds a solution to a prize-collecting instance which captures the problem of pushing back the forwarding time to some future time $t'$. In this instance, the requests are the eligible requests for $\lambda$, and the penalty for a request $q$ is the amount by which its investment counter $h_q$ must be raised so that $q$'s future residual delay would be $0$ at time $t'$. The forwarding time, as well as the corresponding prize-collecting solution, are returned by the call to the function \ForwardTime. 

If the solution returned by $\ForwardTime$ does not serve any requests (i.e. it only raises investment counters), the service modifies it to serve some arbitrary eligible request. While this does not affect the approximation ratio of the algorithm, it bounds the number of services by the number of requests, which bounds the running time of the algorithm.

Now, the algorithm increases the investment counter of eligible requests which are not served by the solution (paying for their future delay until the forwarding time). The algorithm also upgrades the level of those requests, in a similar way to the deadline algorithm.

Finally, the service transmits its solution, serving the remainder of the eligible requests.

\DontPrintSemicolon
\begin{algorithm}
	\renewcommand{\algorithmcfname}{Algorithm}
	\caption{\label{alg:FWY_Algorithm}Network Design with Delay Framework}
	\algorithmfootnote{\footnotemark[\value{footnote}] For the sake of the algorithm and its analysis, no requests outside $Q'_\lambda$ are considered served by this transmission.}
	\BlankLine
	
%
	
	\BlankLine
	
	\EFn{\UponRequest{$q$}}{
		
		Set $ S_q \gets \OP(\{q\})$
		
		Set $ I_q \gets \frac{c(S_q)}{\gamma} $. 
		
		Set $\rank_q \gets \left\lfloor\log \left(I_q\right)\right\rfloor$ \tcp*[h]{the level of the request}
		
	}
	
	\BlankLine
	
	\EFn(\tcp*[h]{Upon a level $j$ becoming critical at time $t$}){\UponCritical{$j$}}{
		
		Start a new service $\lambda$, which we now describe.
		
		Set $\rank_\lambda \gets j + 1$. \label{line:FWY_SetServiceLevel}
		
		\BlankLine
		
		\ForEach(\tcp*[h]{Clean residual delay of eligible requests}){request $q$ such that $\rank_q \le \rank_\lambda$}{Set $h_q \gets h_q + \rho_q(t)$}	 \label{line:FWY_CleanEligibleDelay}	
		
		\BlankLine
		
		Set $E_0 = \left\{ e\in \E \middle| c(e) \le \frac{2^{\rank_\lambda}}{|\E|} \right\}$.	\tcp*[h]{Buy all cheap elements} \label{line:FWY_BuyCheapEdges}
		
		\BlankLine
		
		\tcp*[h]{Forward time}
		
		Let $Q_\lambda$ be all pending requests of level at most $\rank_\lambda$.
		
		Set $(\tau,S) \gets \ForwardTime(E_0, Q_\lambda,\rank_\lambda)$.
		
		Let $Q'_\lambda \subseteq Q_\lambda$ be the subset of requests served in $S$.
		
		\BlankLine
		
		\tcp*[h]{make sure that the service serves at least one pending request}
		
		\lIf{$Q'_\lambda = \emptyset$}{\label{line:FWY_ForceService}for an arbitrary $q\in Q_\lambda$, set $Q'_\lambda\gets \{q\}$ and $S \gets S_q$.}
		
		\BlankLine
		
		\tcp*[h]{pay for future delay of requests unserved by the transmission, and upgrade requests}
		
		\ForEach(){$q\in Q_\lambda\backslash Q'_\lambda$} {
			
			Raise $h_q$ by $\pi_{t\to \tau}(q)$.\label{line:FWY_MakeInvestment}
			
			Set $\rank_{q} \gets \rank_\lambda$. \label{line:FWY_ServiceSetsRequestLevel}
			
		}
		
		\BlankLine
	
		Transmit the solution $E_0 \cup S$, serving the requests $Q'_\lambda$.\footnotemark \label{line:FWY_PerformService}
		
	}

	\BlankLine

\end{algorithm}
	
\begin{algorithm}
	\renewcommand{\algorithmcfname}{Procedure}
	\caption{Time Forwarding Procedure}
	
	\tcc*[h]{This function, called at time $t$, returns a future time $t''$ and a solution $S\subseteq \E$ to transmit which is a ``good'' solution to minimize the future delay of $Q_{\lambda}$ until time $t''$. See Proposition \ref{prop:FWY_TimeForwardingGuarantee} for the formal guarantee of this function.}
	
	\Fn{\ForwardTime{$E_0$, $Q_\lambda$, $j$}}{
	
		Set $t'\gets t$, $Q'_\lambda\gets \emptyset$ and $S \gets \emptyset$. \label{line:FWY_InitSolution}
		
		\While{$Q_\lambda \backslash Q'_\lambda \neq \emptyset$}{
			Let $t'' > t$ be the time in which $\sum_{q\in Q_\lambda \backslash Q'_\lambda}(\pi_{t\to t''}(q)-\pi_{t\to t'}(q))$ reaches $\gamma\cdot 2^{j}$. 
			
			Set $S' \gets \PCOP_{E_0\gets 0}(Q_\lambda, \pi_{t\to t''})$.
			
			\lIf{$c(S') \ge \gamma\cdot 2^{j}$}{\Break}
			
			Set $Q'_\lambda\subseteq Q_\lambda$ to be the set of requests served in $S'$.
			
			Set $t' \gets t''$ and $S \gets S'$.
		}
		
		\Return{$(t'',S)$}
	
	}
\end{algorithm}

\subsection{Analysis}

As in the deadline case, we first consider some definitions and properties of the algorithm before delving into the proof of Theorem \ref{thm:FWY_Competitiveness}.

\subsubsection*{\emph{Definitions and Algorithm's Properties}}

Let $\lambda$ be a service which occurs at some time $t$, making a call to $\ForwardTime(E_0, Q_\lambda, j)$. This call returns the time $\tau$ and a solution $S$ for $\PCOP_{E_0 \gets 0}(Q_\lambda, \pi_{t_\lambda\to \tau})$, where $\pi_\tau$ is as defined in $\lambda$. We prove the following property. 

\begin{prop}
	\label{prop:FWY_TimeForwardingGuarantee}
	The time $\tau$ and solution $S$ returned by $\ForwardTime$ have the following properties:
	\begin{enumerate}
		\item  The cost of $S$ as a solution to $\PCOP_{E_0 \gets 0}(Q_\lambda, \pi_{t \to \tau})$ is at most $2\gamma\cdot 2^j$.
		
		\item Either $S$ serves all requests in $Q_\lambda$ \textbf{or} $\PCOP^*_{E_0 \gets 0}(Q_\lambda, \pi_{t \to \tau}) \ge 2^j$.
	\end{enumerate}
\end{prop}

\begin{proof}
	To prove the first property, consider the final values of the variables $t'$ and $t''$ in \ForwardTime, where the final value of $t''$ is the returned time $\tau$. Observe that the function maintains that $S$ has a cost of at most $\gamma \cdot 2^j$ as a solution for $\PCOP_{E_0 \gets 0}(Q_\lambda, \pi_{t \to t'}$. 
	
	Observing the lines in which the final values of $t'$ and $t''$ were set, we have one of two cases. In the first case, in which $t''=t'$, we are done. Otherwise, we have that $\sum_{q\in Q_\lambda \backslash Q'_\lambda} (\pi_{t \to t''}(q) - \pi_{t \to t'}(q)) = \gamma \cdot 2^j$. In words, the total penalty increase for the requests not served in $S$ from $\pi_{t\to t'}$ to $\pi_{t\to t''}$ is $\gamma\cdot 2^j$. Thus, the solution $S$ has a total cost of at most $2\gamma \cdot 2^j$, proving the first property.
	
	As for the second property, consider the loop of \ForwardTime. If it finishes through the loop's condition, $S$ serves all requests in $Q_\lambda$ and we are done. Otherwise, the loop is ended by the \Break command, in which case we know that the cost of $S'$ as a solution to $\PCOP_{E_0 \gets 0}(Q_\lambda, \pi_{t \to t''})$ is at least $\gamma \cdot 2^j$. But since $S'$ is a $\gamma$ approximation for this problem, we have that $\PCOP^*_{E_0 \gets 0}(Q_\lambda,\pi_{t\to t''}) \ge 2^j$, completing the proof.
\end{proof}

For every service $\lambda$, we denote by $t_\lambda$ the time in which $\lambda$ occurred. In the running of $\lambda$, consider time $\tau$ as returned by \ForwardTime. We call this time the \emph{forwarding time} of $\lambda$, and denote it by $\tau_\lambda$. We call the value set to $\rank_{\lambda}$ the \emph{level} of $\lambda$; observe that this value does not change once defined.

Similarly, for a request $q$, we call $\rank_q$ the level of $q$. Note that unlike services, the level of a request may change over time (more specifically, the level can be increased).

We redefine some of the definitions we used in the deadline case to fit the delay case.

\begin{defn}[Service Pointer]
	Let $q$ be a request. We define $\chrg_q$ to be the last service $\lambda$ such that $\lambda$ sets $\rank_q \gets \rank_\lambda$ in Line \ref{line:FWY_ServiceSetsRequestLevel}. If there is no such service, we write $\chrg_q = \nul$. Similarly, we define $\chrg_q(t)$ to be the last service $\lambda$ before time $t$ such that $\lambda$ sets $\rank_q \gets \rank_\lambda$ in Line \ref{line:FWY_ServiceSetsRequestLevel} (with $\chrg_q(t)=\nul$ if there is no such service).
\end{defn}

\begin{defn}
	\label{defn:FWY_EligibleRequest}
	Consider a service $\lambda$ and a request $q$ which is pending upon the start of $\lambda$, and has $\rank_q \le \rank_\lambda$ at that time. We say that $q$ was \emph{eligible} for $\lambda$. 
\end{defn}

In the algorithm, the set of eligible requests for a service $\lambda$ is the value of the variable $Q_\lambda$. We use this notation throughout the analysis, denoting the set of requests eligible for a service $\lambda$ by $Q_\lambda$.

\begin{defn}
	\label{defn:FWY_ServiceTypes}
	For a service $\lambda$:
	\begin{enumerate}
		\item We say that $\lambda$ is \emph{charged} if there exists some future service $\lambda'$, which is triggered by some level $j$ becoming critical, and there exists a pending request $q$ which is of level $j$ and has positive residual delay immediately before $\lambda'$, such that $\chrg_{q} (t_{\lambda'}) = \lambda$. We say that $\lambda'$ charged $\lambda$.
		
		\item We say that $\lambda$ is \emph{perfect} if the solution $S$ returned by $\ForwardTime$ serves all of $Q_\lambda$. Otherwise, we say that $\lambda$ is \emph{imperfect}.
		
		\item We say that $\lambda$ is \emph{primary} if, when triggered upon $\rank_\lambda -1$ becoming critical, every pending request $q$ of level exactly $\rank_\lambda-1$ with positive residual delay has $\chrg_q (t_\lambda) = \nul$. Otherwise, $\lambda$ is \emph{secondary}.
	\end{enumerate}
\end{defn}

Fix any input set of requests $Q$. We denote by $\Lambda$ the final set of services by the algorithm. We denote the set of primary services made by the algorithm by $\Lambda_1$, and the set of secondary services by $\Lambda_2$, such that $\Lambda = \Lambda_1 \cup \Lambda_2$. We denote the set of charged services by $\Lambda^\circ$.

The algorithm explicitly maintains the following invariant.
\begin{inv}
	\label{inv:FWY_SubcriticalInvariant}
	At any point $t$ during the algorithm, for every set of pending requests $Q'$ of level at most $j$, it holds that $\rho_{Q'}(t) \le 2^j$.
\end{inv}

The following observation is ensured by Lines \ref{line:FWY_PerformService} and \ref{line:FWY_MakeInvestment}.
\begin{obs}
	\label{obs:FWY_NoResidualDelayUntilForwardTime}
	Let $\lambda$ be a service, and let $q$ be a request eligible for $\lambda$. Then $q$ has no residual delay between $t_\lambda$ and $\tau_\lambda$.
\end{obs}

\begin{prop}
	\label{prop:FWY_UniqueCharge}
	Each service is charged by at most one service.
\end{prop}
\begin{proof}
	Assume for contradiction that there exists a service $\lambda$ at time $t$ which is charged by both $\lambda_1$ and $\lambda_2$, at times $t_1$ and $t_2$ respectively, and assume without loss of generality that $t_1 <t_2$. Service $\lambda_2$ charged $\lambda$ due to the pending request $q_2$, such that $\rank_{q_2} = \rank_\lambda$ and $\chrg_{q_2}(t_{\lambda_2}) = \lambda$. $q_2$ was pending before both $\lambda$ and $\lambda_2$, and was thus pending before $\lambda_1$. But after $\lambda_1$, all pending requests are of level at least $\rank_{\lambda_1} = \rank_\lambda + 1$, in contradiction to having $\rank_{q_2} = \rank_\lambda$ immediately before $\lambda_2$.
\end{proof}

\begin{prop}
	\label{prop:FWY_ChargingOnlyAfterForwardTime}
	Suppose a service $\lambda\in \Lambda^\circ$ is charged by a service $\lambda'$. Then $t_{\lambda'} \ge \tau_\lambda$.
\end{prop}
\begin{proof}
	Suppose for contradiction that $t_{\lambda'} < \tau_\lambda$. Denote the level of service $\lambda$ by $j$. The service $\lambda'$ must be triggered by level $j$ becoming critical. Let $Q'$ be the set of requests of level at most $j$ with positive residual delay immediately before $t_{\lambda'}$. Since $\lambda'$ charged $\lambda$, there must be a request $q\in Q'$ such that $\chrg_q (t_{\lambda'}) = \lambda$. Thus, $q$ was eligible for $\lambda$. But thus Observation \ref{obs:FWY_NoResidualDelayUntilForwardTime} contradicts $q\in Q'$.
\end{proof}

\begin{lem}
	\label{lem:FWY_UniqueClass}
	Let $Q'$ be an set of requests, and let $r_{Q'} = \max_{q\in Q'} r_q$. Let $\Lambda$ be the set of charged services for which a request from $Q'$ was eligible and such that for every $\lambda\in \Lambda$ we have $\tau_\lambda \ge r_{Q'}$. Then for every $j\in\mathbb{Z}$, there exists at most one service $\lambda\in \Lambda$ such that $\rank_{\lambda}=j$.
\end{lem}
\begin{proof}
	Assume for contradiction that there exists $j\in\mathbb{Z}$ for which
	there exist two distinct services $\lambda_{1},\lambda_{2}\in\Lambda$ such that $\rank_{\lambda_{1}}=\rank_{\lambda_{2}}=j$. Assume without loss of generality that $t_{\lambda_1}<t_{\lambda_2}$.
	
	Let $\lambda'$ be the service that charged $\lambda_1$. The service $\lambda'$ must be a level $j+1$ service.
	
	Consider the two following cases:
	\begin{enumerate}
		\item  $t_{\lambda'}>t_{\lambda_2}$. Since $\lambda'$ charged $\lambda$, there must be a request $q$ such that $\rank_q = \rank_{\lambda_1}$ and $\chrg_q (t_{\lambda'}) = \lambda_1$. Since $\chrg_q (t_{\lambda'}) = \lambda_1$, we have that $q$ was eligible for $\lambda_1$. Thus, since $t_{\lambda_1} < t_{\lambda_2} < t_{\lambda'}$, $q$ was pending at $\lambda_2$. Since the levels of requests can only increase over time, it must be that $\rank_q \le \rank_{\lambda_1} = \rank_{\lambda_2}$ immediately before $t_{\lambda_2}$. But then $q$ was eligible for $\lambda_2$, and thus $\lambda_2$ would call Line \ref{line:FWY_SetServiceLevel} on $q$, in contradiction to having $\chrg_q (t_{\lambda'})= \lambda_1$.
		
		\item $t_{\lambda'} < t_{\lambda_2}$. Using Proposition \ref{prop:FWY_ChargingOnlyAfterForwardTime}, we know that $t_{\lambda'} \ge \tau_{\lambda_1}$. Since $\lambda_1 \in \Lambda$, we thus have that $t_{\lambda'} \ge t_{Q'}$. Now, consider all pending requests of $Q'$ before $\lambda_2$. Since $t_{Q'} \le t_{\lambda'} < t_{\lambda_2}$, these requests were also pending before $\lambda'$. Since after $\lambda'$ all pending requests are of level at least $\rank_{\lambda'} = j+1$, none of these requests are eligible for $\lambda_2$. This is in contradiction to $\lambda_2 \in \Lambda$.
	\end{enumerate}
	This concludes the proof.
\end{proof}

\subsubsection*{\emph{Upper-bounding $\alg$.}}

\begin{prop}
	\label{prop:FWY_DelayBoundedByCounters}
	The total delay cost of the algorithm is at most $\sum_{q\in Q} h_q$, for the final values of the counters $\{h_q\}_{q\in Q}$.
\end{prop}	
\begin{proof}
	Consider a request $q$, served in some service $\lambda$ at time $t$. Since $q$ was served in $\lambda$, we know that $\rank_q \le \rank_\lambda$ at $t$. From Line \ref{line:FWY_CleanEligibleDelay}, we know that the service $\lambda$ raised $h_q$ so that the residual delay of $q$ becomes $0$. After this line, $h_q$ is at least $d_q(t)$. Since $q$ is served in $\lambda$, its delay does not increase further. 
\end{proof}

To bound the cost of the algorithm, it is thus enough to bound the total cost of transmissions plus the sum of the final values of $h_q$ over requests $q\in Q$. 

We define the cost of a service $\lambda$, denoted by $c(\lambda)$, as the sum of the cost of the transmission made in that service and the total amount by which $\sum_{q\in Q} h_q$ is raised in that service. From Proposition \ref{prop:FWY_DelayBoundedByCounters}, we know that $\sum_{\lambda\in \Lambda} c(\lambda)$ is an upper bound to the cost of the algorithm. We denote this sum by $\widehat{\alg}$.

\begin{lem}
	\label{lem:FWY_ALGUpperBound}
	$\widehat{\alg} \le O(\gamma) \cdot \left(\sum_{\lambda\in \Lambda_1} 2^{\rank_\lambda} + \sum_{\lambda\in \Lambda^\circ} 2^{\rank_\lambda}\right)$
\end{lem}

\begin{prop}
	\label{prop:FWY_ServiceCostBoundedByLevel}
	The total cost of a service $\lambda$ is at most $O(\gamma)\cdot 2^{\rank_\lambda}$.
\end{prop}
\begin{proof}
	The cost incurred in $\lambda$ is at most the sum of the following costs:
	
	\begin{enumerate}
		\item The cost of raising the investment counters at Line \ref{line:FWY_CleanEligibleDelay}, which is at most $2^{\rank_\lambda}$ (using Invariant \ref{inv:FWY_SubcriticalInvariant}).
		
		\item The cost of transmitting the elements $E_0$ in Line \ref{line:FWY_PerformService}, which is  at most $2^{\rank_\lambda}$. 
		
		\item The added cost of transmitting $S$ in Line \ref{line:FWY_PerformService} (given that the transmission already contains $E_0$), and the cost of raising investment counters of requests by $\pi{t \to \tau}$ in Line \ref{line:FWY_MakeInvestment}. Observe that this cost is in fact the cost of $S$ as a solution for $\PCOP_{E_0 \gets 0}(Q_\lambda, \pi_{t\to \tau})$. Since $S$ was obtained from a call to $\ForwardTime(E_0, Q_\lambda, \rank_{\lambda})$, and using Proposition \ref{prop:FWY_TimeForwardingGuarantee}, we have that this cost is at most $2\gamma\cdot 2^{\rank_\lambda}$.
		
%

		\item The cost of the possible transmission in Line \ref{line:FWY_ForceService}. The transmission is of $S_q$, for a request $q$ which is eligible for $\lambda$. Thus, we know that the cost of the transmission is at most $2\gamma\cdot 2^{\rank_\lambda}$. 
	\end{enumerate}
	
	Overall, the costs sum to $O(\gamma)\cdot 2^j$, as required.
\end{proof}

In a perfect service, all eligible requests are served. Thus, Line \ref{line:FWY_ServiceSetsRequestLevel} is never called in a perfect service. The next observation follows.

\begin{obs}
	\label{obs:FWY_OnlyFullServicesCharged}
	Only imperfect services can be charged.
\end{obs}

\begin{proof}[Proof of Lemma \ref{lem:FWY_ALGUpperBound}]
	Observe that $\widehat{\alg} = c(\Lambda_1) + c(\Lambda_2)$. First, observe that through Proposition \ref{prop:FWY_ServiceCostBoundedByLevel} we have that $c(\Lambda_1) \le O(\gamma) \cdot \sum_{\lambda\in \Lambda_1} 2^{\rank_\lambda}$. 
	
	It remains to show that $c(\Lambda_2)\le O(\gamma)\cdot \sum_{\lambda\in \Lambda^\circ} 2^{\rank_\lambda}$. Observe that every secondary service $\lambda$ of level $j$ charges a previous service $\lambda'\in \Lambda^\circ$ of level $j-1$, which is imperfect by Observation \ref{obs:FWY_OnlyFullServicesCharged}. From Proposition \ref{prop:FWY_ServiceCostBoundedByLevel}, we have that $c(\lambda) \le O(\gamma)\cdot 2^j$, and thus $c(\lambda) \le O(\gamma)\cdot 2^{\rank_{\lambda'}}$. Summing over all secondary services completes the proof, where Proposition \ref{prop:FWY_UniqueCharge} guarantees that no charged service is counted twice.
\end{proof}

\subsubsection*{\emph{Lower-bounding $\opt$.}}
Fix the set of services $\Lambda^*$ made in the optimal solution. To complete the proof of Theorem \ref{thm:FWY_Competitiveness}, we require the following two lemmas which lower-bound the cost of the optimal solution.

\begin{lem}
	\label{lem:FWY_PrimaryLowerBoundsOPT}
	$\sum_{\lambda\in \Lambda_1} 2^{\rank_\lambda} \le O(1)\cdot \opt$
\end{lem}

\begin{lem}
	\label{lem:FWY_ChargeLowerBoundsOPT}
	$\sum_{\lambda\in \Lambda^\circ} 2^{\rank_\lambda} \le O(\log n) \cdot \opt$
\end{lem}

\begin{proof}[Proof of Lemma \ref{lem:FWY_PrimaryLowerBoundsOPT}]
	Consider a service $\lambda\in \Lambda_1$ of level $j$. $\lambda$ is triggered upon level $j-1$ becoming critical. Let $Q^\crit_\lambda$ be the set of requests with positive residual delay of level at most $j-1$ which triggered $\lambda$. Define $\sigma_\lambda$ to be the earliest release time of a request in $Q^\crit_\lambda$.
	
	Fix any level $j$. We claim that the intervals of the form $[\sigma_\lambda,t_\lambda]$ for every $j$-level service $\lambda\in \Lambda_1$ are disjoint. Assume otherwise, that some $ [\sigma_{\lambda_1},t_{\lambda_1}] $ and $ [\sigma_{\lambda_2},t_{\lambda_2}] $  intersect. Without loss of generality, assume that $t_{\lambda_1} \in [\sigma_{\lambda_2},t_{\lambda_2}]$. Then there exists a request $q \in Q^\crit_{\lambda_2}$ which was pending during $\lambda_1$, after which $\rank_q$ would be at least $j$, in contradiction to $q\in Q^\crit_{\lambda_2}$.
	
	Now, define $Q^=_\lambda\subseteq Q^\crit_\lambda$ to be the subset of requests in $Q^\crit_\lambda$ which are of level exactly $\rank_\lambda - 1$. Denote by $t_\lambda^-$ the time $t_\lambda$ immediately before the service $\lambda$. Using Invariant \ref{inv:FWY_SubcriticalInvariant}, we have that $\rho_{Q^\crit_\lambda \backslash Q^=_\lambda}(t_\lambda^-) \le 2^{\rank_\lambda-2} $. Thus, we have that $\rho_{Q^=_\lambda}(t_\lambda^-)\ge 2^{\rank_\lambda-2}$. In addition, since $\lambda \in \Lambda_1$, we have that $\chrg_q(t_\lambda) = \nul$ for every $q\in Q^=_\lambda$. Thus, $I_q$ as defined in $\UponRequest$ is at least $2^{\rank_q} = 2^{\rank_\lambda - 1}$. 
	
	Observe that according to the definition of $I_q$, and the approximation guarantee of $\OP$, we have that $I_q$ is a lower bound to the cost of any solution which serves $q$. Thus, we have that during the interval $[\sigma_\lambda,t_\lambda]$ the optimal solution has either served a request from $Q^=_\lambda$ (at a cost of at least $2^{\rank_\lambda-1}$), or paid a delay of $2^{\rank_\lambda-2}$ for the requests of $Q^=_\lambda$.
	
	Now, let $m_j$ be the number of primary services of level $j$, and let $j_{\max}$ be the maximum level of a primary service. Denoting $x^+ = \max(x,0)$, consider the optimal solution. It must pay at least  $2^{j_{\max} -2}$ in either delay or service for each of the $m_{j_{\max}}$ intervals of the form $[\sigma_\lambda,t_\lambda]$ (for $\lambda\in \Lambda_1$ of level $j_{\max}$). For each such service $\lambda$, we charge the optimal solution $2^{j_{\max}-2}$ either for its delay or for a single service in the corresponding interval in which a request from $Q^=_\lambda$ was served.
	
	Now, consider the next level $j_{\max}-1$. We know that the optimal solution must incur $2^{j_{\max} -3}$ for each of the $m_{j_{\max} -1}$ intervals of this level. However, the optimal solution might already be charged for a service of level $j_{\max}$, and might use this service to save costs, serving an interval with cost less than $2^{j_{\max} -3}$. But this can only happen $m_{j_{\max}}$ times, and can only hit a single interval of level $j_{\max} -1$
	(since those intervals are disjoint). Thus, we can charge at least $(m_{j_{\max}-1} -m_{j_{\max}})^+$ intervals an amount of $2^{j_{\max} -3}$, either for delay or for a single service of a level-$(j_{\max}-2)$ request.
	
	Repeating this argument, we get that the optimal solution pays at least $\left(m_j - \max_{j'>j} \{m_{j'} \} \right)^+ \cdot 2^{j-2}$ for each level $j$. 
	
	As for the cost of the algorithm, we have that 
	\begin{align*}
	c(\Lambda_1) &\le O(1)\cdot \sum_{j=-\infty}^{j_{\max}} m_j \cdot 2^j \\
	& \le  O(1)\cdot \sum_{j=-\infty}^{j_{\max}} \left(m_j - \max_{j'>j}\{m_{j'}\} \right)^+ \cdot 2^{j+1} \\
	& \le O(1)\cdot \opt
	\end{align*}
	where the first inequality uses Proposition \ref{prop:FWY_ServiceCostBoundedByLevel} and the second inequality is through changing the order of summation and summing a geometric series.
\end{proof}

It remains to prove lemma \ref{lem:FWY_ChargeLowerBoundsOPT} by charging for each service $\lambda \in \Lambda^\circ$ the amount $2^{\rank_\lambda}$ to the optimal solution times $O(\log |\E|)$. As in the deadline case, we split the charge of $2^{\rank_\lambda}$ between the services made by the optimal solution, and show that each charge is locally valid. 

For a service $\lambda^* \in \Lambda^*$ of the optimal solution, we denote by $Q_{\lambda^*}$ the set of requests served by $\lambda^*$. We define the cost associated with $\lambda^*$, denoted by $c(\lambda^*)$, to be the transmission cost of $\lambda^*$ plus the total delay cost of the requests $Q_{\lambda^*}$ in the optimal solution. Recall that for a service $\lambda \in \Lambda$ made by the algorithm, $Q_\lambda$ is the set of requests eligible for $\lambda$. We define $Q_{\lambda \cap \lambda^*} = Q_\lambda \cap Q_{\lambda^*}$. 

For a set of requests $Q'$, we denote the cost of the optimal offline solution for $\PCOP$ on $Q'$, with respect to a penalty function $\pi:Q'\to \mathbb{R}^+$, by $\PCOP^*(Q',\pi)$. We also use $\PCOP^*_{E_0 \gets 0}(Q', \pi)$ to refer to the cost of the optimal offline solution for $Q'$ where the costs of the elements $E_0 \subseteq \E$ is set to $0$. We also write $\PCOP^*(Q',\pi)$ where $\pi$ is defined on a \emph{superset} of $Q'$; the penalty function in this case is the restriction of $\pi$ to $Q'$.

For a service $\lambda \in \Lambda$, we denote by $E_0^\lambda$ the value set to $E_0$ in Line \ref{line:FWY_BuyCheapEdges} during the service $\lambda$. The outline of the proof of Lemma \ref{lem:FWY_ChargeLowerBoundsOPT} is shown in Figure \ref{fig:FWY_ChargingToOptimal}.

\begin{figure}[tb]
	\subfloat[\label{subfig:FWY_ChargingToOptimal_Charging}Charging Scheme]{\includegraphics[width=0.45\textwidth]{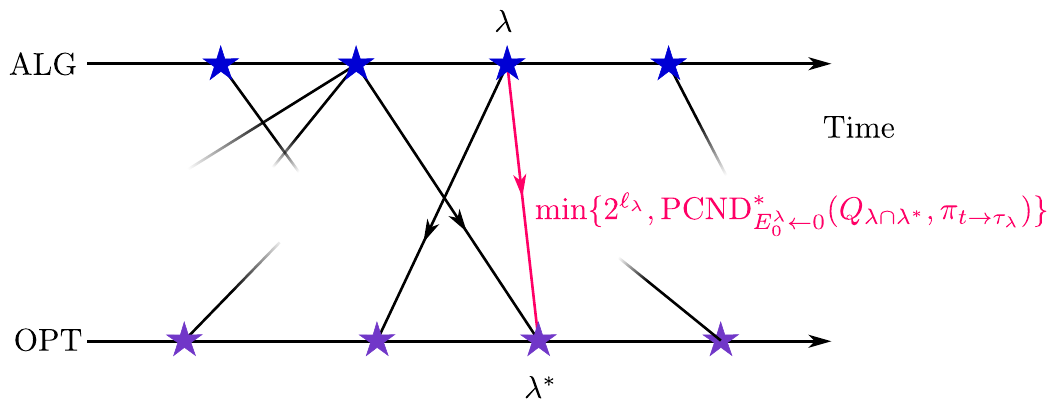}}
	\hfill
	\subfloat[\label{subfig:FWY_ChargingToOptimal_Sink}Charges to Optimal Service]{\includegraphics[width=0.45\textwidth]{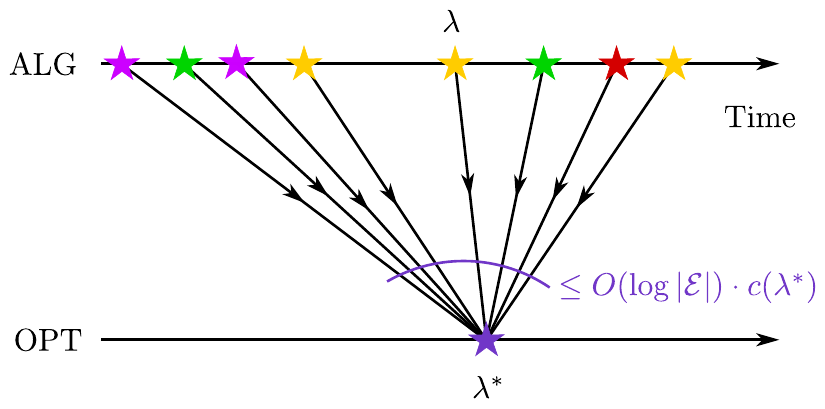}}\\

	\quad In a similar way to Subfigure \ref{subfig:FWD_ChargingToOptimal_Charging}, Subfigure \ref{subfig:FWY_ChargingToOptimal_Charging} shows the services of $\Lambda^\circ$ and the services of the optimal algorithm, as well as the charging of costs to the optimal solution. The amount $\min\{2^{\rank_\lambda}, \PCOP^*_{E_0^\lambda\gets 0}(Q_{\lambda\cap\lambda^*}, \pi_{t_\lambda \to \tau_\lambda})\}$ is charged by the service $\lambda \in \Lambda^\circ$ to the optimal service $\lambda^*$. The proof of Lemma \ref{lem:FWY_ChargeLowerBoundsOPT} shows that these charges are sufficient, i.e. each service $\lambda\in \Lambda^\circ$ charges at least $2^{\rank_\lambda}$.
	
	\quad Subfigure \ref{subfig:FWY_ChargingToOptimal_Sink} shows the validity of the charging, given in Proposition \ref{prop:FWY_LocalCharge}. As in the deadline case, this proposition shows that the total amount charged to an optimal service $\lambda^*$ exceedes its cost by a factor of at most $O(\log |\E|)$. The argument is similar to Proposition \ref{prop:FWD_LocalCharge}. However, in addition to the three types of services in the deadline case (green, yellow, red), there is an additional type of service (pink), which consists of services $\lambda$ with $\tau_\lambda \le t_{\lambda^*}$. These pink services are shown to charge a total of at most $c(\lambda^*)$.
	
	\caption{\label{fig:FWY_ChargingToOptimal} Visualization of Services}
\end{figure}


\begin{prop}
	\label{prop:FWY_LocalCharge}
	There exists a constant $\beta$ such that for every optimal service $\lambda^* \in \Lambda^*$, we have that 
	\begin{equation}
	\label{eq:FWY_ChargePerOptimalService}
	\sum_{\lambda\in \Lambda^\circ}\min\{ 2^{\rank_\lambda},\PCOP^*_{E_0^\lambda \gets 0}(Q_{\lambda \cap \lambda^*}, \pi_{t_\lambda \to \tau_\lambda}) \} \le \beta\log |\E| \cdot c(\lambda^*)
	\end{equation}
\end{prop}

\begin{proof}
	Fix any service $\lambda^* \in \Lambda^*$ of the optimal solution. Observe that a service $\lambda \in \Lambda^\circ$ such that $Q_\lambda \cap Q_{\lambda^*} = \emptyset$ does not contribute to the left-hand side of Equation \ref{eq:FWY_ChargePerOptimalService}. Hence, it remains to consider only $\lambda \in \Lambda^\circ$ such that $Q_\lambda \cap Q_{\lambda^*} \neq \emptyset$; denote the set of such services by $\Lambda'$.
	
	Define $t^* = \max_{q\in Q_{\lambda^*}} r_q$. Each $\lambda \in \Lambda'$ is in one of the following cases.
	
	\paragraph*{Case 1: $\tau_\lambda \le t^*$.} Let $\Lambda^{\le t^*} \subseteq \Lambda'$ be the subset of such services. For every request $q$ eligible for $\lambda$, define $h_q^\lambda$ to be the value of the investment counter $h_q$ upon the start of $\lambda$. We have:
	
	\begin{align*}
	\sum_{\lambda \in \Lambda^{\le t^*} }\min\{ 2^{\rank_\lambda},\PCOP^*_{E_0^\lambda \gets 0}(Q_{\lambda \cap \lambda^*},\pi_{t_\lambda \to \tau_\lambda}) \} &\le  \sum_{\lambda \in \Lambda^{\le t^*} } \PCOP^*_{E_0^\lambda \gets 0}(Q_{\lambda \cap \lambda^*},\pi_{t_\lambda \to \tau_\lambda}) \\
	& \le \sum_{\lambda \in \Lambda^{\le t^*} } \sum_{q\in Q_{\lambda \cap \lambda^*}}\pi_{t_\lambda \to  \tau_\lambda}(q) \\
	& = \sum_{\lambda \in \Lambda^{\le t^*} } \sum_{q\in Q_{\lambda \cap \lambda^*}} \max\{0,d_{q}(\tau_\lambda) - h_q^\lambda\}\\
	& = \sum_{q\in Q_{\lambda^*}}\sum_{\lambda \in \Lambda^{\le t^*} | q\in Q_\lambda} \max\{0,d_{q}(\tau_\lambda) - h_q^\lambda\}
	\end{align*}
	
	Now, fix any request $q\in Q_{\lambda^*}$. We claim that $\sum_{\lambda\in \Lambda^{\le t^*} | q\in Q_\lambda} \max\{0,d_q(\tau_\lambda) - h_q^\lambda\} \le d_q(t^*)$. To see this, consider the services in the sum by order of occurrence, denoted $\lambda_1,\cdots, \lambda_l$. We prove by induction that $\sum_{i'=0}^i \max\{0,d_q(\tau_{\lambda_{i'}}) - h_q^{\lambda_{i'}}\} \le d_q(t^*)$ for every $i\in [l]$, which proves the claim. Clearly, this holds for the base case of $i=1$, since $\max\{0,d_q(\tau_{\lambda_{1}}) - h_q^{\lambda_{1}}\} \le d_q(\tau_{\lambda_{1}}) \le d_q(t^*)$. 
	
	We prove the inductive claim for $i>1$ by assuming it holds for $i-1$. Observe that $\lambda_1, \cdots, \lambda_{i-1}$ paid the penalty for $q$ (otherwise it would not be eligible for $\lambda_i$). Thus, we have that at the end of $\lambda_{i-1}$ we have that $h_q \ge \sum_{i'=0}^{i-1} \max\{0,d_q(\tau_{\lambda_{i'}}) - h_q^{\lambda_{i'}}\} \le d_q(t^*)$. Since $h_q^{\lambda_i}$ can only be larger, and since $\max\{0,d_q(\tau_{\lambda_{i}}) - h_q^{\lambda_{i}}\} \le d_q(t^*) - h_q^{\lambda_i}$, the inductive claim holds.
	
	Overall, for this case, we have that 
	\[\sum_{\lambda \in \Lambda^{\le t^*} }\min\{ 2^{\rank_\lambda},\PCOP^*_{E_0^\lambda \gets 0}(Q_{\lambda \cap \lambda^*}) \}\le \sum_{q\in Q_{\lambda^*}} d_q(t^*) \le c(\lambda^*) \]
	where the last inequality is due to the fact that $\lambda^*$ occurs no earlier than $t^*$, and thus the optimal solution incurs the delay of $Q_{\lambda^*}$ up to $t^*$.
	
	\paragraph*{Case 2: $\tau_\lambda > t^*$.} Denote by $\Lambda^{>t^*} \subseteq \Lambda'$ the set of such services. Using Lemma \ref{lem:FWY_UniqueClass}, for every level $j$ there exists at most one $j$-level service in $\Lambda^{>t^*}$. Define $\ell = \lfloor \log (c(\lambda^*_i)) \rfloor$, and consider the following subcases for $\lambda \in \Lambda^{>t^*}$:
	\begin{enumerate}
		\item $\rank_\lambda \le \ell$. In this case, we have that $\lambda$ contributes at most $2^{\rank_\lambda}$ to the left-hand side of Equation \ref{eq:FWY_ChargePerOptimalService}. Summing over at most a single service from each level yields a geometric sum which is at most $2^{\ell+1} \le 2\cdot c(\lambda^*)$.
		
		\item $\ell < \rank_\lambda < \ell + \lceil \log |\E|  \rceil+ 1$. For such $\lambda$, observe that 
		\[\min\{2^{\rank_\lambda}, \PCOP^*_{E_0^\lambda \gets 0}(Q_{\lambda \cap \lambda^*}, \pi_{t_\lambda \to \tau_\lambda})\} \le \OP^*(Q_{\lambda^*}) \le c(\lambda^*) \]
		and thus the service $\lambda$ contributes at most $c(\lambda^*)$ to the left-hand side of Equation \ref{eq:FWY_ChargePerOptimalService}. Summing over at most one $\lambda$ from each level, their total contribution to the left-hand side of Equation \ref{eq:FWY_ChargePerOptimalService} is at most $\lceil \log |\E| \rceil \cdot c(\lambda^*)$.

		\item $\rank_\lambda \ge \ell + \lceil \log |\E| \rceil +1$. We claim that $\PCOP^*_{E_0^\lambda \gets 0}(Q_{\lambda \cap \lambda^*}) = 0$, and thus the contribution to the left-hand side of Equation \ref{eq:FWY_ChargePerOptimalService} from these services is $0$. 
		
		To prove this claim, observe that $\PCOP^*_{E_0^\lambda \gets 0}(Q_{\lambda \cap \lambda^*}, \pi_{t_\lambda \to \tau_\lambda}) \le \OP^*_{E_0^\lambda \gets 0}(Q_\lambda^*)$. Consider that every element in $\lambda^*$ costs at most $c(\lambda^*) \le 2^{\ell +1}$. Thus, since $2^{\rank_\lambda} \ge 2^{\ell+1} \cdot |\E|$, we have that $\lambda$ added all elements of $\lambda^*$ to $E_0$ in Line \ref{line:FWY_BuyCheapEdges}. Note that since $\lambda^*$ served $Q_{\lambda^*}$, we have that $\OP^*_{E_0^\lambda \gets 0}(Q_\lambda^*) =0$, as required.
	\end{enumerate}
	Summing over the contributions from each level completes the proof.
\end{proof}

\begin{proof}[Proof of Lemma \ref{lem:FWY_ChargeLowerBoundsOPT}]
	As in the deadline case, it is enough to show that for every charged service $\lambda \in \Lambda^\circ$, we have that 
	\begin{equation}
		\label{eq:FWY_GlobalChargeIsLocalCharge}
		2^{\rank_\lambda} \le \sum_{\lambda^*\in \Lambda^*} \min\{2^{\rank_\lambda}, \PCOP^*_{E_0^\lambda \gets 0}(Q_{\lambda \cap \lambda^*}, \pi_{t_\lambda \to \tau_\lambda})\}
	\end{equation}
	Summing over all $\lambda \in \Lambda^\circ$ and using Proposition \ref{prop:FWY_LocalCharge} would immediately yield the lemma.
	
	If one of the summands on the right-hand side of Equation \ref{eq:FWY_GlobalChargeIsLocalCharge} is $2^{\rank_\lambda}$, the claim clearly holds, and the proof is complete. Otherwise, the right-hand side is exactly $\sum_{\lambda^* \in \Lambda^*} \PCOP^*_{E_0^\lambda \gets 0}(Q_{\lambda \cap \lambda^*}, \pi_{t_\lambda \to \tau_\lambda})$. Now, since $\bigcup_{\lambda^* \in \Lambda^*} Q_{\lambda \cap \lambda^*} = Q_\lambda$, we can construct a feasible solution for $\PCOP_{E_0^\lambda\gets 0}(Q_\lambda, \pi_{t_\lambda \to \tau_\lambda})$ by buying the elements in $\PCOP^*_{E_0^\lambda\gets 0}(Q_{\lambda \cap \lambda^*}, \pi_{t_\lambda \to \tau_\lambda})$ for every $\lambda^* \in \Lambda^*$, and paying the penalty for unserved requests. Clearly, the cost of this solution is at most $\sum_{\lambda^* \in \Lambda^*} \PCOP^*_{E_0^\lambda \gets 0}(Q_{\lambda \cap \lambda^*}, \pi_{t_\lambda \to \tau_\lambda})$, and thus 
	
	\[ \PCOP^*_{E_0^\lambda\gets 0}(Q_\lambda, \pi_{t_\lambda \to \tau_\lambda}) \le \sum_{\lambda^* \in \Lambda^*} \PCOP^*_{E_0^\lambda \gets 0}(Q_{\lambda \cap \lambda^*}, \pi_{t_\lambda \to \tau_\lambda})  \]
	
	From Observation \ref{obs:FWY_OnlyFullServicesCharged}, we know that $\lambda$ is an imperfect service. Proposition \ref{prop:FWY_TimeForwardingGuarantee} thus implies that  $2^{\rank_\lambda}   \le \PCOP^*_{E_0^\lambda\gets 0}(Q_\lambda, \pi_{t_\lambda \to \tau_\lambda})$, which completes the proof. 
	
\end{proof}

\begin{proof}[Proof of Theorem \ref{thm:FWY_Competitiveness}]
	The competitiveness guarantee results immediately from Lemmas \ref{lem:FWY_ALGUpperBound}, \ref{lem:FWY_ChargeLowerBoundsOPT} and  \ref{lem:FWY_PrimaryLowerBoundsOPT}.
	
	As for the running time of the algorithm, it is clear that it is determined by either Line \ref{line:FWY_BuyCheapEdges}, which takes $O(|\E|)$ time in each service, or by the number of calls made to the prize-collecting approximation algorithm $\PCOP$ in the function \ForwardTime. We claim that the total number of calls made in each service is $O(k^2)$, with $k = |Q|$ the number of requests in the input. 
	
	To see this, fix any service $\lambda$ at time $t$. Observe that the number of calls made to $\PCOP$ in a service $\lambda$ is exactly the number of iterations of the loop in \ForwardTime. Denote the iterations of this loop in the service $\lambda$ by $I_1,\cdots, I_l$. For every iteration $I_i$, we denote by $t_i$ the value of the variable $t''$ set in iteration $I_i$, and denote by $S_i$ the $\PCOP$ solution computed in $I_i$. 
	
	Observe the state after iteration $i_k$ -- we know that the requests of $Q_\lambda$ gather a total delay of at least $k\gamma\cdot 2^{\rank_\lambda}$ between $t$ and $t_k$. Thus, there exists a request $q_1\in Q_\lambda$ which has delay of at least $\gamma\cdot 2^{\rank_\lambda}$. In any solution $S_i$ for $i>k$ (except possibly the final one $S_l$), we have that $q$ is served. This is since otherwise the cost of $S_i$ would exceed $\gamma\cdot 2^{\rank_\lambda}$, in contradiction to the loop not ending at the \Break command in \ForwardTime.
	
	Next, consider the iterations $i_{k+1},\cdots, i_{2k}$. Using the same argument, we know that there exists a request $q_2 \in Q_\lambda \backslash \{q_1\}$ that gathers at least $\gamma\cdot 2^{\rank_\lambda}$ delay until time $t_{2k}$. Thus, $S_i$ for $2k\le i < l$ serves $q_2$. Repeating this argument, we know that for $i\ge k^2$ the solution $S_i$ must serve all requests $Q_\lambda$, ending the loop. 
	
	Note that the number of services performed in the algorithm is at most $k$, since each service serves some pending request (as ensured by Line \ref{line:FWY_ForceService}). Thus, the total running time consists of $O(k^3)$ calls to $\PCOP$, and $O(k |\E|)$ time for Line \ref{line:FWY_BuyCheapEdges}. This completes the proof.
\end{proof}

\section{Applications and Extensions of the Delay Framework}
In this section, we apply the framework of Section \ref{sec:FWY} to various problems, as we did for the deadline case. The requirement for the delay framework is an approximation algorithm for the prize-collecting problem. For some of the problems we consider, we cite appropriate prize-collecting algorithms. For others, we use a simple construction which yields a prize-collecting approximation algorithm from an approximation algorithm for the original problem.

\paragraph{Edge-Weighted Steiner Tree and Forest.} The following result is due to Hajiaghayi and Jain~\cite{Hajiaghayi2006}. 

\begin{pthm}[\cite{Hajiaghayi2006}]
	There exists a polynomial-time, deterministic $3$-approximation for EW prize-collecting Steiner forest.
\end{pthm}

Plugging the algorithm of the previous theorem into the framework of Section \ref{thm:FWY_Competitiveness} yields the following result.

\begin{thm}
	There exists an $O(\log n)$-competitive deterministic algorithm for EW Steiner forest with delay which runs in polynomial time.
\end{thm}

\paragraph{Multicut.} The result of Garg \textit{et al.}~\cite{Garg1993}, stated in Theorem \ref{thm:APP_MCD_GargEtAl}, is in fact an approximation with respect to the optimal fractional solution for the following LP relaxation (where $\P_q$ is the collection of paths connecting the two terminals of $q$).

\begin{equation}
\label{eq:MCD_PrimalLP}
\arraycolsep=5pt \def\arraystretch{2.2}
\begin{array}{lcr}
\text{minimize} & \sum_{e\in E}x_e c(e)& \\
\text{subject to} & \sum_{e\in P} x_e \ge 1 & \forall q\in Q, \in \P_q\\
& x_e \ge 0 &\forall e\in E
\end{array}
\end{equation}

The corresponding prize-collecting LP relaxation, for a penalty function $\pi$, is the following.

\begin{equation}
\label{eq:MCD_PC_PrimalLP}
\arraycolsep=5pt \def\arraystretch{2.2}
\begin{array}{lcr}
\text{minimize} & \sum_{e\in E}x_e c(e) + \sum_{q \in Q} p_q \pi(q)& \\
\text{subject to} & \sum_{e\in P} x_e + p_q \ge 1 & \forall q\in Q, \in \P_q\\
& x_e \ge 0 &\forall e\in E
\end{array}
\end{equation}

The following construction is a folklore construction of a prize-collecting approximation algorithm from an approximation algorithm for the original problem. First, we solve the prize-collecting LP in Equation \ref{eq:MCD_PC_PrimalLP} to obtain a solution $\left( \{x_e\}_{e\in E},\{p_q\}_{q\in Q} \right)$. For each request $q$ such that $p_q \ge \frac{1}{2}$ the algorithm pays the penalty. The remainder of the requests are solved by calling the approximation algorithm for the original (non-prize-collecting) problem. This construction can easily be seen to lose only a constant factor (namely, 2) over the approximation ratio of the original approximation algorithm. 

For the case of multicut, first observe that this construction is indeed implementable -- that is, the prize-collecting LP can be solved in polynomial time by using a classic separation oracle based on min-cut queries for each request. Thus, the resulting approximation guarantee for the construction is $O(\log n)$. Plugging the resulting algorithm into the framework of Section \ref{sec:FWY} yields the following result.

\begin{thm}
	There exists a deterministic $O(\log^2 n)$-competitive algorithm for multicut with delay which runs in polynomial time.
\end{thm}

\paragraph{Node-Weighted Steiner Forest.} The following result is due to Bateni \textit{et al.}~\cite{DBLP:journals/siamcomp/BateniHL18}.

\begin{pthm}[\cite{DBLP:journals/siamcomp/BateniHL18}]
	There exists a polynomial time, deterministic $O(\log n)$-approximation for node-weighted prize-collecting Steiner forest.
\end{pthm}

Plugging the algorithm of the previous theorem into the framework of Section \ref{thm:FWY_Competitiveness} yields the following result.

\begin{thm}
	There exists an $O(\log^2 n)$-competitive deterministic algorithm for EW Steiner forest with delay which runs in polynomial time.
\end{thm}

\paragraph{Edge-Weighted Steiner Network.} The following result is due to Hajiaghayi and Nasri~\cite{DBLP:conf/latin/HajiaghayiN10}.

\begin{pthm}[\cite{DBLP:conf/latin/HajiaghayiN10}]
	There exists a polynomial-time, deterministic $3$-approximation for EW prize-collecting Steiner network.
\end{pthm}

Plugging the algorithm of the previous theorem into the framework of Section \ref{thm:FWY_Competitiveness} yields the following result.

\begin{thm}
	There exists an $O(\log n)$-competitive deterministic algorithm for EW Steiner network with delay which runs in polynomial time.
\end{thm}

\paragraph{Directed Steiner Tree}

The recent result of Grandoni \textit{et al.}~\cite{Grandoni:2019:OAA:3313276.3316349} for directed Steiner tree is based on an approximation algorithm to a problem called Group Steiner Tree on Trees with Dependency Constraint (GSTTD), which they show is equivalent to directed Steiner forest. Their algorithm for GSTTD is an approximation with respect to the optimal solution to a rather complex LP relaxation, which involves applying Sherali-Adams strengthening to a base relaxation for GSTTD.

At the time of writing this paper, we could not find a consideration of the prize-collecting variant of directed Steiner tree. We conjecture that a construction similar to shown here for Steiner forest would also apply for directed Steiner tree, yielding a prize-collecting algorithm with only a constant loss in approximation over the original algorithm of~\cite{Grandoni:2019:OAA:3313276.3316349}. 

While proving the existence of such a component is beyond the scope of this paper, we nonetheless state the resulting guarantee for directed Steiner tree with delay assuming that the component exists.

\begin{thm}
	If there exists a $\gamma$-approximation for prize-collecting directed Steiner tree which runs in quasi-polynomial time, then there exists an $O(\gamma \log n)$-competitive algorithm for directed Steiner tree with delay which also runs in quasi-polynomial time.
\end{thm}

\subsection{Facility Location}

The following result is due to Xu and Xu~\cite{DBLP:journals/jco/XuX09}.

\begin{pthm}
	\label{thm:APY_FLY_XuAndXu}[\cite{DBLP:journals/jco/XuX09}]
	There exists a polynomial-time, deterministic $1.8526$-approximation for prize-collecting facility location.
\end{pthm}

In this subsection we prove the following result.

\begin{thm}
	\label{thm:APY_FLY_Competitiveness}
	There exists a deterministic $O(\log n)$-competitive algorithm for facility location with delay.
\end{thm}

As previously observed in the deadline case, the facility location problem does not conform to the $\OP$ structure, and thus the framework cannot be applied to facility location in a black-box fashion and still obtain $O(\log n)$ loss. In the deadline case, we showed that the framework of Section \ref{sec:FWD} could still be directly applied to facility location; the only necessary modification was in the analysis -- namely, the proof of Lemma \ref{lem:FWD_ChargeLowerBoundsOPT}.

In facility location with delay, however, this is not the case -- a minor modification to the framework itself is required. The modification is simply to ensure that during any ongoing service, the investment counter of a pending request never surpasses the cost of connecting that request to an open facility. 

The modification consists of replacing the \textbf{foreach} loop of Line \ref{line:FWY_MakeInvestment} with the modification in Snippet \ref{alg:APY_FLY_Algorithm}.

\begin{algorithm}
	\renewcommand{\algorithmcfname}{Snippet}
	\caption{\label{alg:APY_FLY_Algorithm} Facility Location Modification}
	
	Let $F$ be the set of facilities opened in $S$.
	
	\ForEach{$q\in Q\backslash Q'_\lambda$}{
		\eIf{$h_q + \pi_{t''}(q) \ge \min_{u\in F}\delta(u,q)$\label{Line:APY_FLY_ServiceIfCondition}}{ 
			Set $h_q = \max(h_q, \min_{u\in F}\delta(u,q))$
			
			Set $Q'_\lambda \gets Q'_\lambda \cup \{q\}$
			
			Modify $S$ to also serve $q$ by connecting $q$ to $\arg \min_{u\in F} \delta(u,q)$.
		}{
			Set $h_q \gets h_q +\pi_{t''}(q)$.
			
			Set $\rank_{q} \gets \rank_\lambda$. 
		}
		
	}
	
\end{algorithm}

As was the case in facility location with deadlines, Remark \ref{rem:APP_FLD_SolutionNature} applies to the nature of solutions in the facility location with delay algorithm.

\subsubsection*{Analysis}

We show that the application of the framework in Section \ref{sec:FWD}, with the modification of Snippet \ref{alg:APY_FLY_Algorithm}, to the approximation algorithm of Theorem \ref{thm:APY_FLY_XuAndXu} proves Theorem \ref{thm:APY_FLY_Competitiveness}. As in the deadline case, we would like to reprove Lemmas \ref{lem:FWY_ALGUpperBound}, \ref{lem:FWY_PrimaryLowerBoundsOPT} and \ref{lem:FWY_ChargeLowerBoundsOPT} for facility location with delay, which would prove the theorem. 

For Lemma \ref{lem:FWY_ALGUpperBound}, consider that the cost of serving additional requests in the snippet is bounded by the investment counters of those requests -- thus, losing a factor of $2$, we ignore this additional cost. The remaining argument is identical to the original proof of Lemma \ref{lem:FWY_ALGUpperBound}.

Lemma \ref{lem:FWY_PrimaryLowerBoundsOPT} goes through without modification. It remains to prove Lemma \ref{lem:FWY_ChargeLowerBoundsOPT} for our case. As in the deadline case, the only part of the proof which needs to be modified is the local-charging proposition, which is Proposition \ref{prop:FWY_LocalCharge}.

\begin{proof}[Proof of Proposition \ref{prop:FWY_LocalCharge} for facility location]
	We use the notation defined in the original proof of Proposition \ref{prop:FWY_LocalCharge}. The proof breaks down in the third subcase of case 2 -- that is, the case of a service $\lambda$ which forwarded past time $t^*$, such that $\rank_\lambda \ge  \ell + \lceil \log |\E| \rceil +1$. Let $\Lambda^\gg$ be the collection of services in this subcase. We claim that 
	\[\sum_{\lambda\in \Lambda^\gg} \PCOP^*_{E_0^\lambda \gets 0}(Q_{\lambda \cap \lambda^*},\pi_{\lambda,\tau_\lambda}) \le 2\cdot c^\conn(\lambda^*) \label{eq:APY_FLY_SubcaseProof}   \] 
	
	where $c^\conn(\lambda^*) \le c(\lambda^*)$ is the connection cost incurred by the optimal solution in $\lambda^*$. To show this, for every $\lambda \in \Lambda^\gg$ we define the following solution $\S$ for $\PCOP_{E_0^\lambda \gets 0}(Q_{\lambda \cap \lambda^*},\pi_{\lambda,\tau_\lambda})$:
	\begin{enumerate}
		\item Open facilities at all nodes in $E_0^\lambda$, at cost $0$.
		
		\item For every request $q \in Q_{\lambda \cap \lambda^*}$:
		\begin{enumerate}
			\item If $\lambda$ is the last service in $\Lambda^\gg$ for which $q$ is eligible, connect $q$ to the closest facility in $E_0^\lambda$.
			
			\item Otherwise, pay the penalty $\pi_{\lambda, \tau_\lambda}(q)$.
		\end{enumerate}
	\end{enumerate}

	This solution has no opening cost, only connection and penalty costs. We now count the costs of those solutions by each request separately, attributing to a request $q\in Q_{\lambda^*}$ the connection and penalty cost incurred for it by the solutions. 
	
	Fix a request $q\in Q_{\lambda^*}$, and denote by $\lambda_1, \cdots, \lambda_l \in \Lambda^\gg$ the services for which $q$ was eligible, ordered by time of occurrence. For every $i \in [l]$, denote by $\S_i$ the solution corresponding to $\lambda_i$. Denote by $E^*$ the set of facilities opened in $\lambda^*$ and observe that, as in the original proof, for every $\lambda_i$ for $i\in [l]$ we have that $E^*\subseteq E_0^{\lambda_i}$. Thus, the total cost due to $q$ is:
	
	\textbf{penalty}: penalty cost $\pi_{\lambda_i,\tau_{\lambda_i}}$ is paid in $\S_i$ for $i$ such that $\lambda_i$ does not serve $q$. The services $\lambda_i$ in which the solution pays the penalty for $q$ do not serve $q$; observe that in such services $h_q$ increases by $\pi_{\lambda_i,\tau_{\lambda_i}}$. After each such $\lambda_i$, we also have that $h_q \le \min_{v\in E_0^{\lambda_i}} \delta(v,q)$ -- otherwise, the \textbf{if} condition in Line \ref{Line:APY_FLY_ServiceIfCondition} in the snippet would force $q$ to be served, in contradiction. In particular, $h_q \le \min_{v\in E^*} \delta(v,q)$ after each such $\lambda_i$. This implies that the sum of penalty costs for $q$ is at most $ \min_{v\in E^*} \delta(v,q)$, which is the connection cost of $q$ in $\lambda^*$.
	
	\textbf{connection:} There exists at most one index $i\in [l]$ such that $\S_i$ connects $q$. Using again the fact that $E^*\subseteq E_0^{\lambda_i}$, the connection cost of request $q$ in $\S_i$ is at most the connection cost of $q$ in $\lambda^*$.
	
	We completes the proof of Equation \ref{eq:APY_FLY_SubcaseProof}. Thus, we have that the contribution from services $\lambda \in \Lambda^\gg$ to the left-hand side of Equation \ref{eq:FWY_ChargePerOptimalService} is at most $2\cdot c(\lambda^*)$, completing the proof of the proposition.
\end{proof}

\subsection{Exponential-Time Algorithms}

As in the deadline case, one can use the framework of Section \ref{sec:FWY} to obtain the following information-theoretic upper bound on competitiveness.

\begin{thm}
	There exists an $O(\log |\E|)$-competitive algorithm for $\OP$ with delay (with no guarantees on running time). In particular, there exists an $O(\log n)$-competitive algorithm for all problems considered in this paper, where $n$ is the number of nodes in the input graph.
\end{thm}

\section{Request-Based Regime}
\label{sec:RBF}
In problems with deadlines or with delay, the usual regime is that the number of requests is unbounded, and potentially much larger than the size of the underlying universe (e.g. the number of nodes in the graph). This is the regime we addressed in this paper thus far. However, for offline network design, the opposite regime is used -- i.e. that the universe is large, and the number of requests is much smaller. For such a regime, it is preferable to give guarantees in the number of requests $k$. In this section, we obtain the best of both worlds, namely a guarantee in the minimum between the number of requests and the size of the universe. The following theorem states the result of this section.

\begin{thm}
	\label{thm:RBF_Competitiveness}
	If there exists a $\gamma$ deterministic (randomized) approximation algorithm for $\OP$, then there exists an $O(\gamma\log (\min\{k,|\E|\}))$-competitive deterministic (randomized) algorithm for $\OP$ with deadlines, which runs in polynomial time.
\end{thm}

\subsection{Proof of Theorem \ref{thm:RBF_Competitiveness}}

To prove Theorem \ref{thm:RBF_Competitiveness}, we first show how to modify the framework of Section \ref{sec:FWD} to be $O(\gamma \log k)$-competitive, where $\gamma$ is the approximation ratio of the encapsulated approximation algorithm. We then describe a simple way to combine this modified framework with the original framework of Section \ref{sec:FWD} to prove Theorem \ref{thm:RBF_Competitiveness}.

\subsubsection*{Modified $O(\gamma \log k)$-Competitive Framework}

We describe the needed modification to the framework of Section \ref{sec:FWD} to achieve $(\gamma \log k)$-competitiveness. For the sake of describing the framework, we assume that the number of requests $k$ is known in advance (this assumption is later relaxed using standard doubling techniques). The single modification required is in the definition of $E_0$, as defined in $\UponDeadline$. Instead of adding all cheap elements (those that cost at most $\frac{2^{\rank_\lambda}}{|\E|}$), we instead iterate over pending requests which are cheap.

Namely, the new framework is obtained by replacing Line \ref{line:FWD_BuyCheapEdges} with Snippet \ref{alg:RBF_E0Redefinition}, which defines $E_0$ in a different way.

\begin{algorithm}
	\renewcommand{\algorithmcfname}{Snippet}
	\caption{\label{alg:RBF_E0Redefinition} Facility Location Modification}
	
	\While{there exists a pending request $q$ which is not served by $E_0$, such that $c(S_q)\le \frac{\gamma\cdot 2^{\rank_\lambda}}{k}$}{

		Set $E_0 \gets E_0 \cup S_q$

	}
	
\end{algorithm}

\subsubsection*{Analysis}
The following theorem states the competitiveness of the modified framework.

\begin{thm}
	\label{thm:RBF_IntermediaryCompetitiveness}
	The framework of Section \ref{sec:FWD}, when modified with Snippet \ref{alg:RBF_E0Redefinition}, is $O(\gamma \log k)$-competitive.
\end{thm}
 
The proof of Theorem \ref{thm:RBF_IntermediaryCompetitiveness} is very similar to the proof of Theorem \ref{thm:FWD_Competitiveness}. Lemma \ref{lem:FWD_ALGUpperBound} goes through in an almost identical way -- it is enough to notice that the cost of $E_0$ as defined in Snippet \ref{alg:RBF_E0Redefinition} never exceeds $\gamma \cdot 2^{\rank_\lambda}$.

Lemma \ref{lem:FWD_PrimaryLowerBoundsOPT} also goes through in an identical manner. It remains to prove the following analogue to Lemma \ref{lem:FWD_ChargeLowerBoundsOPT}.

\begin{lem}[Analogue of Lemma \ref{lem:FWD_ChargeLowerBoundsOPT}]
	\label{lem:RBF_ChargeLowerBoundsOPT}
	$\sum_{\lambda\in \Lambda^\circ} 2^{\rank_\lambda} \le O(\log k) \cdot \opt$
\end{lem}

To prove Lemma \ref{lem:RBF_ChargeLowerBoundsOPT}, we only need to prove the following analogue of Proposition \ref{prop:FWD_LocalCharge}. The proof of Lemma \ref{lem:RBF_ChargeLowerBoundsOPT} from this analogue is identical to the proof of Lemma \ref{lem:FWD_ChargeLowerBoundsOPT} from Proposition \ref{prop:FWD_LocalCharge}.

\begin{prop}[Analogue of Proposition \ref{prop:FWD_LocalCharge}]
	There exists a constant $\beta$ such that for every optimal service $\lambda^* \in \Lambda^*$, we have that 
	\begin{equation}
		\label{eq:RBF_ChargePerOptimalService}
		\sum_{\lambda\in \Lambda^\circ}\min\{2^{\rank_\lambda},\OP^*_{E_0^\lambda\gets 0}(Q_{\lambda \cap \lambda^*})\} \le \beta\log k \cdot c(\lambda^*)
	\end{equation}
\end{prop}

\begin{proof}
	The proof is very similar to the proof of Proposition \ref{prop:FWD_LocalCharge}. Fix an optimal service $\lambda^* \in \Lambda^*$. Denote by $\Lambda'\subseteq \Lambda^\circ$ the subset of charged services made by the algorithm in which a request from $Q_{\lambda^*}$ is served (other services, for which $Q_{\lambda \cap \lambda^*}=\emptyset$, need not be considered). Observe that $Q_{\lambda^*}$ is an intersecting set, as the optimal solution served $Q_{\lambda^*}$ is a single point in time. Lemma \ref{lem:FWD_UniqueClass} implies that for every level $j$, there exists at most one $j$-level service in $\Lambda'$. Define $\ell = \lfloor \log (c(\lambda^*)) \rfloor$.  Now, consider the following cases for a service $\lambda\in \Lambda'$:
	\begin{enumerate}
		\item $\rank_\lambda \le \ell$. Each such $\lambda$ contributes at most $2^{\rank_\lambda}$ to the left-hand side of Equation \ref{eq:FWD_ChargePerOptimalService}. Summing over at most one service from each level yields a geometric sum which is at most $2^{\ell +1} \le 2\cdot c(\lambda^*)$.
		
		\item $\ell < \rank_\lambda < \ell + \lceil \log k  \rceil+ 1$. For such $\lambda$, observe that $\min\{2^{\rank_\lambda}, \OP^*_{E_0^\lambda \gets 0}(Q_{\lambda \cap \lambda^*})\} \le \OP^*(Q_\lambda) \le c(\lambda^*)$. Summing over at most a single service from each level, the total contribution to the left-hand side of Equation \ref{eq:FWD_ChargePerOptimalService} from these levels is at most $\lceil \log k \rceil\cdot c(\lambda^*)$.
		
		\item $\rank_\lambda \ge \ell + \lceil \log k \rceil +1$. Observe that $\min\{2^{\rank_\lambda}, \OP^*_{E_0^\lambda \gets 0}(Q_{\lambda \cap \lambda^*})\} \le \OP^*_{E_0^\lambda \gets 0}(Q_{\lambda \cap \lambda^*})$. We now claim that $\OP^*_{E_0^\lambda \gets 0}(Q_{\lambda \cap \lambda^*}) =0$, which implies that the total contribution from these levels to the left-hand side of Equation \ref{eq:RBF_ChargePerOptimalService} is $0$. 
		
		Indeed, consider that $\OP^*(\{q\}) \le c(\lambda^*)$ for every request $q\in Q_{\lambda^*}$ (since $\lambda^*$ is itself a feasible solution). If, in addition, we have that $q\in Q_\lambda$, then $q$ was pending immediately before $\lambda$. From the approximation guarantee of $\OP$, we have that $c(S_q) \le \gamma \cdot \OP^*(\{q\}) \le \gamma \cdot c(\lambda^*) \le \gamma \cdot 2^{\ell+1}$. Thus, since $2^{\rank_\lambda} \ge 2^{\ell+1}\cdot k$, Snippet \ref{alg:RBF_E0Redefinition} guarantees that $E_0^\lambda$ serves $q$. Since this holds for every $q\in Q_{\lambda \cap \lambda^*}$, we have that $\OP^*_{E_0^\lambda \gets 0}(Q_{\lambda \cap \lambda^*})=0$.
	\end{enumerate}
	Summing over the contributions from each level completes the proof.
\end{proof}

\begin{proof}[Proof of Theorem \ref{thm:RBF_IntermediaryCompetitiveness}]
	The proof of the theorem results immediately from Lemmas \ref{lem:FWD_ALGUpperBound}, \ref{lem:FWD_PrimaryLowerBoundsOPT} and \ref{lem:RBF_ChargeLowerBoundsOPT}. The analysis of the running time remains the same.
\end{proof}

\subsubsection*{Proof of Theorem \ref{thm:RBF_Competitiveness}}

First, we describe the doubling we use to relax the assumption that $k$ is known to the algorithm. We do this by guessing a value $\hat{k}$ for the number of requests -- initially a constant -- and running the framework of Theorem \ref{thm:RBF_IntermediaryCompetitiveness} for that value. When the number of requests exceeds $\hat{k}$, we send all new requests to a new instance of the algorithm (which is run in parallel to the previous instances), in which the guessed number of requests is $\hat{k}^2$. We then set $\hat{k}\gets \hat{k}^2$. 

The cost of the $i$'th instance is at most $\gamma \log \hat{k}_i \cdot \opt$, where $\hat{k}^i$ is the value of $\hat{k}$ used by the $i$'th instance. Consider that the final instance is that in which $\hat{k} \ge k$, and that for this instance we have $\hat{k}\le k^2$ and thus $\log \hat{k} \le 2\log k$. Since $\log \hat{k}$ grows by a factor of $2$ with each iteration, we have that the total cost of the algorithm is at most $4\gamma \log k\cdot \opt$, as required.

To prove Theorem \ref{thm:RBF_Competitiveness}, we modify this by stopping the doubling process earlier: when $\hat{k}$ exceeds $|\E|$, we start a new instance of the original framework of Section \ref{sec:FWD}, and send all new requests to that instance. This is easily seen to achieve the desired competitiveness bound.

\paragraph{Extension to Delay.}

The modifications seen in this section for deadlines can also be applied to the delay framework of Section \ref{sec:FWY}, achieving an identical guarantee to Theorem \ref{thm:RBF_Competitiveness}. However, as is the case in the original delay framwork, we cannot allow a pending request which is not eligible to the current service to be served by this service -- otherwise, Proposition \ref{prop:FWY_DelayBoundedByCounters} would no longer hold, as the residual delay of an ineligible request might be nonzero. This yields the following result.

\begin{thm}

	\label{thm:RBF_DelayCompetitiveness}
	If there exists a $\gamma$ deterministic (randomized) approximation algorithm for $\PCOP$, then there exists an $O(\gamma\log (\min\{k,|\E|\}))$-competitive deterministic (randomized) algorithm for $\OP$ with delay, which runs in polynomial time.

\end{thm}

\subsection{Applications}

We can apply this framework to the network design problems which conform to the structure of $\OP$. In Section \ref{sec:APP}, we chose to quote the approximation ratios of all offline approximation algorithms in terms of $n$ instead of $k$, since we were interested in a guarantee in $n$ (the reader can verify that the original guarantees of these algorithms are indeed in terms of $k$). 

In this section, we are interested in a guarantee in $\min\{k,n\}$. We thus replace $n$ with $\min\{n,k\}$ in the approximation ratios of all offline approximation algorithms stated in Section \ref{sec:APP}. Plugging those approximation algorithms into the framework, Theorem \ref{thm:RBF_Competitiveness} yields the following results:

\begin{table}[h!]
	\begin{center}
		\caption{Framework Applications}
		\label{tab:RBF_ResultsTable}
		\begin{tabular}{l|c} 
			Edge-weighted Steiner forest with deadlines & $O(\log \min\{k,n\})$\\
			Multicut & $O(\log^2 \min\{k,n\})$ \\
			Edge-weighted Steiner network & $O(\log \min\{k,n\})$ \\
			Node-weighted Steiner forest &$O(\log^2 \min\{k,n\})$ \\		
			Directed Steiner tree & $O\left(\frac{\log^3 \min\{k,n\}}{\log \log \min\{k,n\}}\right)$ \\			
			
		\end{tabular}
	\end{center}
\end{table}

\section{Conclusions and Open Problems}

This paper presented frameworks for network design problems with deadlines or delay, which encapsulate approximation algorithms for the offline network design problem, with competitiveness which is a logarithmic factor away from the approximation ratio of the underlying approximation algorithm. The running time of these frameworks has a polynomial overhead over the running time of the encapsulated approximation algorithm. 

In particular, in the formal online model with unbounded computation, this provides $O(\log n)$ upper bounds (with $n$ the number of vertices in the graph), when the offline problem is solved exactly. For some network design problems, as seen in Appendix \ref{sec:LB}, this is relatively tight -- that is, an information-theoretic lower bound of $\Omega(\sqrt{\log n})$ exists. Whether there exists an improved framework which can bridge this gap remains open.

For the remaining network design problems, the gap is still large, as no non-constant lower bound is known. This raises the possibility of designing a framework which works for a restricted class of network design problems (which excludes node-weighted Steiner tree and directed Steiner tree), but yields constant competitiveness results for this restricted class. Either designing such a framework, or showing lower bounds, is an open problem.

An additional open problem is to design a good approximation for prize-collecting directed Steiner tree. Applying Theorem \ref{thm:FWY_Competitiveness} to such a result would yield a competitive algorithm for directed Steiner tree with delay.

\bibliographystyle{plain}
\bibliography{bibfile}
\appendix

\section{Lower Bounds}
\label{sec:LB}
Some of the more difficult network design problems considered in this paper -- namely, node-weighted Steiner tree and directed Steiner tree -- have an information-theoretic lower bound of $\Omega(\sqrt{\log n})$ on competitiveness. This lower bound stems from containing the set cover with delay problem (denoted SCD), first presented in \cite{DBLP:conf/latin/CarrascoPSV18}. 

\begin{thm}
	\label{thm:LB_LowerBound}
	Every randomized algorithm for node-weighted Steiner tree with deadlines (or delay) or directed Steiner tree with deadlines (or delay) has a competitive ratio of $\Omega(\sqrt{\log n})$.
\end{thm}

In the set cover with delay problem, $n'$ elements and $m'$ sets are given. Requests arrive on the elements over time, each with an associated delay function. At any point in time, the algorithm may transmit a set $S$ at a cost $c(S)$, serving all pending requests on elements in the set $S$.

In \cite{DBLP:journals/corr/abs-1807-08543}, a lower bound was presented for set cover with delay, which also applies to deadlines (as all requests in this lower bound construction can be replaced with deadline requests).  Specifically, they gave for every $i$ an instance of SCD in which:
\begin{enumerate}
	\item The number of elements is $n'=3^i$.
	
	\item The number of sets is $m'=2^i$.
	
	\item The competitiveness of any randomized algorithm is at least $\Omega(\sqrt{i})$.
\end{enumerate}

Now, we use standard reductions from set cover to either node-weighted Steiner tree or directed Steiner tree, both on a graph of $n = n'+m'+1$ vertices. The reductions are shown in Figure \ref{fig:LB_Reduction}. Using the lower bound for SCD, we have that $i=\Omega(\log n)$, and thus the competitive ratio of any randomized algorithm is $\Omega(\sqrt{\log n})$, proving Theorem \ref{thm:LB_LowerBound}.
\begin{figure}
	\caption{\label{fig:LB_Reduction}Reduction from Set Cover to Node-Weighted Steiner Tree}
	\includegraphics*{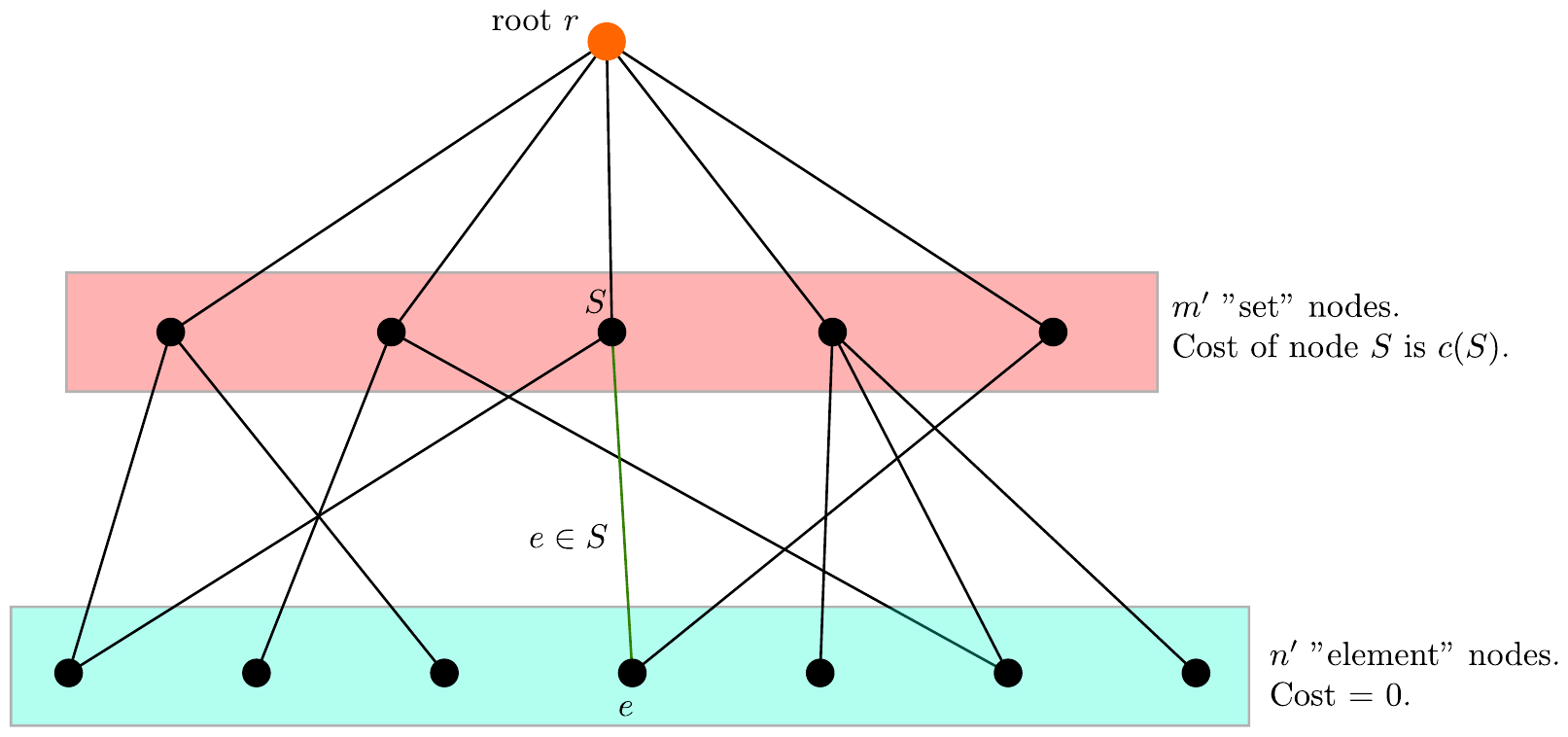}
	
	This figure describes a node-weighted Steiner tree graph of $n'+m'+1$ nodes formed from a set cover instance with $m'$ sets and $n'$ elements. In this graph, the root is connected to $m'$ nodes corresponding to the sets of the set cover instance. There are also $n'$ nodes corresponding to the elements of the instance. Each "set" node is connected to the "element" nodes corresponding to elements in the set. The cost of each set node is exactly the cost of the set in the set cover instance; the cost of the remainder of the nodes is $0$. The reduction from SCD to node-weighted Steiner tree with deadlines consists of translating a request on an element to a request on the corresponding element node. 
	
	The reduction of set cover to directed Steiner tree is similar -- the only differences are that the edges are now directed downward, and that the costs are on the edges from the root to the sets instead of on the set nodes themselves.	
\end{figure}

%
%
%
%
%
%
%
%
%
%
%

\end{document}